\documentclass[runningheads]{llncs}
\usepackage{tikz}
\usepackage{listings}
\usepackage{xcolor}
\usepackage{xspace}
\usepackage{multirow}
\usepackage[textsize=tiny]{todonotes}
\usepackage{amsmath}
\usepackage{stmaryrd}
\usepackage{url}
\usepackage{paralist}
\usepackage[hidelinks]{hyperref}
\hypersetup{bookmarksdepth=2}
\usepackage[capitalize]{cleveref}
\usepackage{colortbl}
\usepackage{fp}
\usepackage{enumitem}
\usepackage{placeins}
\usepackage{float}
\usepackage{stfloats}
\usepackage[frozencache,cachedir=minted-cache]{minted}
\usepackage{amssymb}
\usepackage{pgfplots}
\usepackage{xfp}
\usepackage[numbers,sort&compress]{natbib}
\usepackage{booktabs}
\usepackage{paralist}
\usepackage{lineno}
\usepackage{wrapfig}
\usepackage{mathtools}
\usepackage{orcidlink}
\usepackage{tabularx}
\usepackage{longtable}

\AddToHook{cmd/appendix/before}{%
  \crefalias{section}{appendix}%
  \crefalias{subsection}{appendix}%
}
\definecolor{codehighlight}{RGB}{220, 255, 220}
\setminted{breaklines, frame=none, framesep=3mm,  %
           mathescape=true, escapeinside=||, highlightcolor=codehighlight,fontsize=\scriptsize} %
\definecolor{operator}{RGB}{0, 102, 187}

\crefname{figure}{Figure}{Figures}
\crefname{table}{Table}{Tables}
\crefname{example}{Example}{Examples}
\crefname{section}{Section}{Sections}
\crefname{definition}{Definition}{Definitions}

\usetikzlibrary{positioning, arrows.meta, matrix, calc}

\definecolor{codegray}{rgb}{0.5,0.5,0.5}
\definecolor{backcolour}{rgb}{0.95,0.95,0.92}

\lstdefinestyle{mystyle}{
    numberstyle=\tiny\color{codegray},
    basicstyle=\ttfamily\scriptsize,
    breakatwhitespace=false,
    keepspaces=true,
    showspaces=false,
    showstringspaces=false,
    numbers=left,
    numbersep=5pt,
    lineskip=0pt,
    aboveskip=0pt,
    belowskip=0pt,
    escapeinside={<@}{@>}
}

\tikzset{   every node/.style={minimum height=5mm, text height=1.5ex, text depth=.25ex},
            level distance=12mm}

\newcommand*{\extended}{}%
\newcommand{\ignore}[1]{}
\newcommand{\Digests}{{\mathcal{A}}}
\newcommand{\NDigests}[1]{{\mathcal{A}_{#1}}}
\newcommand{\D}{\mathcal{D}}

\newcommand{\sem}[1]{\llbracket #1 \rrbracket}
\newcommand{\semT}[1]{\llbracket #1 \rrbracket}

\newcommand{\Let}{\textbf{let}}

\newcommand{\In}{\textbf{in}}
\newcommand{\aSem}[1]{\left\llbracket #1 \right\rrbracket^\sharp}

\newcommand{\aSemDigest}[1]{\left\llbracket #1 \right\rrbracket^\sharp_{\Digests}}
\newcommand{\aSemNDigest}[2]{\left\llbracket #2 \right\rrbracket^\sharp_{\Digests_{#1}}}

\newcommand{\pthreadOnce}{\texttt{pthread\_once}\xspace}

\newcommand{\creationLockset}{creationLockset\xspace}
\newcommand{\descendantLockset}{descendantLockset\xspace}

\newcommand{\option}{\xspace\textbf{ opt}}
\newcommand{\oSome}{\textbf{Some }}
\newcommand{\oNone}{\textbf{None}\xspace}

\newcommand{\act}{\textsf{act}}
\newcommand{\create}{\textsf{create}}

\newcommand{\join}{\textsf{join}}
\newcommand{\lock}{\textsf{lock}}
\newcommand{\initMT}{\textsf{initMT}}
\newcommand{\unlock}{\textsf{unlock}}
\newcommand{\return}{\textsf{return}}

\newcommand{\initActRaw}{initMT}
\newcommand{\initAct}{\textsf{\initActRaw}}
\newcommand{\Pos}{\textsf{Pos}}
\newcommand{\Neg}{\textsf{Neg}}

\newcommand{\Nodes}{\mathcal{N} }
\newcommand{\Edges}{\mathcal{E} }
\newcommand{\Actions}{\mathcal{A}ct}
\newcommand{\Traces}{{\mathcal{T}}}

\newcommand{\TIDs}{\textsf{TID}}
\newcommand{\pTID}{\pi_{\textsf{tid}}}
\newcommand{\Barriers}{\mathcal{B}}

\newcommand{\Mutexes}{\mathcal{S}}
\newcommand{\Onces}{\mathcal{O}}

\newcommand{\onceStart}{\textsf{startO}}
\newcommand{\onceEnd}{\textsf{endO}}
\newcommand{\onceInit}{\textsf{initO}}
\newcommand{\onceRan}{\textsf{ran}}

\newcommand{\unique}{\textsf{unique}}

\newcommand{\Bool}{\textsf{Bool}}
\newcommand{\maintid}{\textsf{main}}

\newcommand{\false}{\textsf{false}}

\newcommand{\maycreate}{\textsf{may\_create}}
\newcommand{\mustanc}{\textsf{must\_anc}}
\newcommand{\semDigests}[1]{\aSem{#1}_\Digests}

\newcommand{\newAbs}{\new^\sharp}
\newcommand{\newDigests}{\new^\sharp_\Digests}
\newcommand{\newNDigests}[1]{\new^\sharp_{\Digests_{#1}}}

\newcommand{\init}{\textsf{init}}
\newcommand{\initAbs}{\init^\sharp}
\newcommand{\initDigests}{\init^\sharp_\Digests}
\newcommand{\initNDigests}[1]{\init^\sharp_{\Digests_{#1}}}

\newcommand{\new}{\textsf{new}}

\newcommand{\pthreads}{\textsc{Pthreads}\xspace}
\newcommand{\Goblint}{\textsc{Goblint}\xspace}
\newcommand{\GoblintAnon}{\textsc{Goblint}\xspace}
\newcommand{\SVCOMP}{\textsc{Sv-Comp}\xspace}

\newcommand{\ltDrawingDefs}{
    \tikzset{
        every node/.style={node distance=40pt and 30pt},
        programpoint/.style={circle,draw,font=\small},
        programpointsink/.style={programpoint,line width=0.5mm},
        programpointwide/.style={programpoint,node distance=40pt and 55pt},
        programpointnarrow/.style={programpoint,node distance=40pt and 15pt},
        programpointwidesink/.style={programpointwide,line width=0.5mm},
        edgelabel/.style={midway,font=\small,above=0.5mm}
    }
    \newcommand{\programo}[3]{
        \draw[->](##1)--(##3)node[edgelabel]{$##2$};
    }
    \newcommand{\programon}[4][]{
        \node[programpoint##1,right= of ##2](##4){};
        \draw[->](##2)--(##4)node[edgelabel]{$##3$};
    }
    \newcommand{\programond}[4][]{
        \node[programpoint##1,right= of ##2,dotted](##4){};
        \draw[->,dotted](##2)--(##4)node[edgelabel]{$##3$};
    }
    \newcommand{\createo}[3][]{
        \draw[-latex,blue](##2)--(##3)node[edgelabel,##1]{$\to_c$};
    }
    \newcommand{\joino}[3][]{
        \draw[-latex,brown](##2)--(##3)node[edgelabel,##1]{$\to_j$};
    }
    \newcommand{\mutexo}[5][]{
        \draw[-latex,red,##5](##2) to node[edgelabel,##1]{$\to_{##4}$} (##3);
    }
    \newcommand{\signalo}[5][]{
        \draw[-latex,green!50!black,##5](##2) to node[edgelabel,##1]{$\to_{##4}$} (##3);
    }
}

\ifdefined\extended
    \newcommand{\appref}[2]{\cref{#2}}
\else
    \newcommand{\appref}[2]{Appendix~#1 of the extended version~\cite{extended-version}}
\fi

\ifdefined\extended
    \newcommand{\apprefInline}[2]{\cref{#2}}
\else
    \newcommand{\apprefInline}[2]{the extended version~\cite[Appendix~#1]{extended-version}}
\fi

\ifdefined\extended
    \newcommand{\apprefInlineShort}[2]{\cref{#2}}
\else
    \newcommand{\apprefInlineShort}[2]{extended version~\cite[Appendix~#1]{extended-version}}
\fi

\begin{document}

\title{Beyond Locks and Thread \emph{ID}s:\\ Static Data Race Detection Off The Beaten Path}
\titlerunning{Static Data Race Detection Off The Beaten Path}

\ifdefined\extended
\subtitle{(Extended Version)}
\fi

\author{
Daniel Bund\inst{1}\orcidlink{0000-0002-8675-5070}\and
Julian Erhard\inst{2}\orcidlink{0000-0002-1729-3925}\and
Michael Petter\inst{1}\orcidlink{0009-0009-7868-651X}\and
Michael Schwarz\inst{3}\orcidlink{0000-0002-9828-0308}
}

\authorrunning{D. Bund, J. Erhard, M. Petter, and M. Schwarz}

\institute{
    Technical University of Munich, Garching, Germany\\
    \email{\{daniel.bund,michael.petter\}@tum.de} \and
     CISPA Helmholtz-Zentrum, Garching, Germany\\
     \email{julian.erhard@cispa.de} \and
     National University of Singapore, Singapore\\
     \email{m.schwarz@nus.edu.sg}
}

\maketitle              %

\begin{abstract}
	Maintaining an abstraction of the execution history of threads can improve the precision of data race detection
	in static analysis.
	Here, we extend the \emph{digest} framework to handle concurrency constructs and synchronization mechanisms
	that have been ignored in static race detection.
	We introduce mechanisms for the commonly used thread barriers, as well as \pthreadOnce, which allows to ensure
	that an action is executed only once.
	We also instantiate the framework with an abstraction of locksets held by ancestor threads.
    We propose a suite of litmus tests to evaluate analyses for these features and compare our
	implementation to state-of-the-art tools, finding that they lack support.
\keywords{data races, concurrency, static program analysis, software verification, abstract interpretation}
\end{abstract}

\section{Introduction}
Data races are a class of errors that is particularly hard to debug, partially due to the nondeterministic nature of the
scheduling, i.e., bugs not materializing in every execution.
To ensure correctness of code, particularly in domains where bugs are unacceptable, \emph{sound} static analyses have been proposed
that are \emph{guaranteed} to flag any potential data race. However, these analyses typically suffer from false positives.
For such tools to be practical, the number of false positives must remain manageable~\cite{ChristakisB16, JohnsonSMB13}.
Nevertheless, much of the literature on static data-race detection considers only mutexes and thread creation and joining as synchronization,
overlooking less common but equally valid ways to prevent conflicting accesses from occurring concurrently.
One such feature are barriers as provided by \textsc{Pthreads}:
\begin{figure}
    \centering
    \footnotesize
    \begin{minipage}[t]{5cm}
\begin{minted}[tabsize=1,linenos,fontsize=\scriptsize]{c}
int g;

pthread_barrier_t b;

void* worker(void* arg) {
  // Worker initialization

  pthread_barrier_wait(&b);

  int f = compute(g,arg); |\label{pthread:barrier:example:compute}|
  // Do something with f

  return NULL;
}

\end{minted}
            \end{minipage}\hspace{2em}\begin{minipage}[t]{5.5cm}
\begin{minted}[tabsize=2,linenos,firstnumber=last, fontsize=\scriptsize]{c}
int main() {
  pthread_t w1, w2; // Worker threads

  pthread_barrier_init(&b, NULL, 3);

  pthread_create(&w1, NULL, worker, "1");
  pthread_create(&w2, NULL, worker, "2");

  // Main initialization code
  g = 2; |\label{pthread:barrier:example:main:init}|

  pthread_barrier_wait(&b);
  // ...
}
\end{minted}
            \end{minipage}

    \caption{Program using \pthreads\ barriers.}\label{f:barrier}
\end{figure}

\begin{example}
	\Cref{f:barrier} is a program fragment using barriers.
	The main thread creates two workers, which compute with the value of a global variable $g$
	and their string argument. %
	The barrier $b$ is used to ensure that both workers can only perform their computation (line~\ref{pthread:barrier:example:compute})
	after the main thread has set the global in line~\ref{pthread:barrier:example:main:init}.
	Once all three threads have called wait for $b$,
	all of them can proceed.
	Assuming that the variable $g$ is only accessed in the three locations given in the program fragment,
	there is no data race:
	The write to $g$ in line~\ref{pthread:barrier:example:main:init} happens before all threads have reached the barrier,
	and the other two accesses to $g$ happen later.\label{e:barrier}
\end{example}

On top of synchronization by barriers, we also consider the \pthreadOnce feature library developers often rely on
(see the example \cref{f:onceComplete} in \cref{ss:once})
and two patterns combining reasoning about threads and synchronization by mutexes
(see the \emph{descendantLocksets} example in \cref{f:descendant-ls-example} in \cref{s:descendant-lockset},
as well as the
\emph{creationLocksets} examples
in \cref{f:creation-ls-example}).

We focus on the temporal aspect of race detection, providing sound methods to prove that two access instructions cannot happen in parallel.
This takes different forms: For barriers, our approach determines that
certain accesses must happen before others. For \pthreadOnce, this is combined with establishing that all accesses
protected by the same \pthreadOnce control variable happen in the same thread and can thus not race, even when their
specific order cannot be established.
When it comes to reasoning about threads and mutexes, in the \emph{descendantLocksets} case (\cref{s:descendant-lockset}),
it is once more must-happen-before relationships that are established, while for \emph{creationLocksets}
(\apprefInlineShort{F}{app:creation-lockset-long}),
it is only the fact that accesses are in fact ordered by critical sections even though they may not appear to be at first glance.
Broadly, these mechanisms fall into two categories: barriers, \pthreadOnce, and descendant locksets use additional
history information to establish must-happen-before relationships between accesses,
whereas creation locksets (\apprefInlineShort{F}{app:creation-lockset-long}) extend conventional lockset reasoning by recognizing
protection through critical sections spanning multiple threads.

Our approach is parametric in the other important ingredient of static race detection, namely a \emph{may-alias} analysis,
and can be combined with different such analyses to obtain different speed--precision tradeoffs.

We base our work on \emph{digest-driven abstract interpretation}. \emph{Abstract interpretation}~\cite{popl77} is a principled way to overapproximate the
concrete semantics of programs using abstract domains. \emph{Digests}~\cite{schwarz2023clustered} were originally designed as a mechanism to improve
the precision of thread-modular value analyses using abstract interpretation.
Digests allow an analysis to distinguish program points as well as the states of shared variables according to the execution history observed by the thread executing
the statement.
They thus generalize trace partitioning~\cite{Mauborgne05} to concurrent programs, providing \emph{concurrency-sensitivity}
that can be configured independently of the analysis domain.
In the result computed by an analysis with digests, potentially reachable accesses to shared variables are annotated with digests, i.e., an abstraction of the execution history observed by the thread in question.
This can be exploited to enhance the precision of sound static data race detection~\cite{SchwarzE26}:
When considering two accesses to a shared variable, both annotated with a digest, data races can be excluded when the two accesses cannot happen concurrently---as judged by their digests.

While previous work~\cite{schwarz2023clustered, hammerandnail, SchwarzE26} proposes digests for history information such as whether the execution is single-threaded or multi-threaded,
or which mutexes are locked by the current thread, libraries such as \textsc{Pthreads} offer further concurrency and synchronization primitives,
which have often been ignored.
In particular, we provide digests for barriers, a mechanism to synchronize multiple threads.
To do so, we extend the semantic framework underlying the digest framework to actions that involve more than two threads (\cref{s:preliminaries}).
We show how digests are used to argue about whether two arbitrary operations may possibly happen in parallel\footnote{While \emph{may happen concurrently} is perhaps the more accurate term,
we use \emph{may happen in parallel} for consistency with the literature on MHP analysis.}, which is not only a key ingredient of our
barrier digest but also paves the way towards further applications of the digest framework.

We show in \cref{ss:barriers} how a history abstraction, i.e., a digest, relating to the operations performed on the barrier, is used to deduce race freedom for the programs using barriers, such as \cref{e:barrier}.
There, the predicate $||^?$, introduced for checking whether two accesses to the same variable may race, is used for a different purpose in the definition of the barrier digest itself,
 to determine whether a sufficient number of threads may wait at the barrier at the same time.

\Cref{ss:once} proposes a digest for \pthreadOnce for more precise data-race detection.
\pthreadOnce is a feature used in particular by library developers to ensure that some code is executed exactly once.

Sometimes a thread locks a mutex before creating child threads and does not release it until certain protected operations have completed.
This pattern requires intricate handling beyond existing per-thread lockset analyses.
Partially inspired by~\citet{Lu2025}, we propose a way of handling such patterns with digests in \cref{s:descendant-lockset}.
Another pattern, where locks are held by a thread during the entire execution of its child is covered in \cref{s:creation-lockset}.
We implemented the proposed digests in an abstract interpreter. The experiments we conducted on a set of hand-crafted examples using these constructs are discussed in \cref{s:experiments}. %

\section{Digest-Driven Data Race Detection}\label{s:preliminaries}

At a high level, digest-driven data race detection~\cite{SchwarzE26} proceeds in two phases.
First, an abstract interpreter propagates abstract program states indexed by \emph{digests},
compact summaries of the execution history observed by the current thread, such as
its thread \emph{id}, currently held locks, or completed synchronization events.
Several digests, each tailored to a different aspect of concurrency, can be active
simultaneously and jointly index the abstract states.
Whenever the analysis reaches a shared-memory access, it records the access together
with its digest. Second, the race detector considers pairs of accesses that may refer
to the same memory location and of which at least one is a write. It applies a
may-happen-in-parallel predicate to each active digest: if any one of them establishes
from the aspect of the histories it records that the accesses cannot occur concurrently,
the pair is discarded; otherwise, it is reported as a potential data race. Thus,
digests provide the concurrency information
used both to partition abstract states during the analysis and to rule out infeasible
races afterward.

We next introduce the local trace semantics underlying the digest framework,
and formally define digests and the may-happen-in-parallel predicate.
The local trace semantics precisely describes concurrent programs
while abstracting away irrelevant details such as the order of local actions of independent threads,
which are differentiated in an interleaving semantics.
It has served as the justification for thread-modular analyses and the digest framework.
\emph{Observable} and \emph{observing} actions play a key role: they may be ordered relative to actions of other threads.
Previous work fixed observing actions to arity two, so features such as \pthreadOnce or barriers could not be expressed.
We directly present our generalization to observing actions of arbitrary arity.
Previous work encoded atomic global accesses using cumbersome \emph{atomicity mutexes}.
We instead allow operations to be both observable and observing, avoiding atomicity mutexes.

\subsection{Local Trace Semantics}\label{ss:lt}

A program is a set of labelled control-flow graphs (CFGs),
one for each thread entry point. Let $\Nodes$ be the CFG nodes (\emph{program points}),
$\Actions$ the actions, and $\Edges \subseteq \Nodes \times \Actions \times \Nodes$ the labelled CFG edges.
Program text can be translated into this representation using standard compiler techniques.

A local trace encodes the history of one \emph{ego} thread --- the thread whose perspective is taken --- and the parts of other threads' histories it can know.
The knowledge comes from \emph{observing} actions, which observe \emph{observable} actions of other threads.
Examples include locking (observing), unlocking (observable), and global-variable accesses (both).
\Cref{t:actions} summarizes observing and observable actions.
Further, there are \emph{local} actions (neither observing nor observable).

\begin{table*}[t]
    \caption{Observable and observing actions for primitives addressed in this paper.
    }\label{t:actions}
    \centering
    \begin{tabular*}{0.9\textwidth}{@{\extracolsep{\fill}}lll@{}}
        \toprule
        Feature & Observable Action & Observing Action \\
        \midrule
        Access to a global & $g=x$ / $x=g$ & $g=x$ / $x=g$ \\
        Mutexes & $\unlock(a)$ & $\lock(a)$ \\
        Joining & $\return\,x$ & $x' {=} \join(x'')$ \\
        Barriers &  $\texttt{barrier\_N\_record(b)}$ & $\texttt{barrier\_N\_check(b)}$ \\
        \pthreadOnce & $\onceEnd(o)$ / $\onceInit(o)$ & $\onceStart(o)$ \\
        \bottomrule
    \end{tabular*}
\end{table*}
A local trace contains one action-interleaved configuration sequence per appearing thread.
Each sequence follows a path through the CFG and is \emph{locally consistent}:
its local-variable reads, writes, and taken guards agree~\cite{schwarzthesis}.
This sequence is the thread's \emph{swim-lane}, ordered by the \emph{program order} $\to_{po}$.
For each non-main thread, a $\to_c$ edge connects its \texttt{create} action to the first node of its swim-lane.
Observing and observable actions induce further dependencies between configurations, possibly across threads,
ensuring global consistency.
For example---given mutex $m$ from the set of mutexes $\Mutexes$---there is a dependency $\to_{m}$ from an
\texttt{unlock(m)} to the next \texttt{lock(m)} operation, where each unlock of $m$ (and the initialization action) has at most one outgoing edge
$\to_m$ and each lock
has exactly one incoming edge $\to_m$ (for the first lock, this edge originates at the endpoint
of a special observable action $\initMT$\footnote{
This program-start action executes first and is observed when a mutex is first locked or a global first accessed.
We elide this purely technical action henceforth.

}). Thus, mutex operations are totally ordered, so at most one thread holds a mutex at a time.
Likewise, a $\to_g$ edge connects each access to global $g$ to the next, encoding atomic global accesses.
We revisit its relation to race reporting at the section's end.
In a local trace, then,
\begin{itemize}
    \item Each swim-lane is locally consistent;
    \item Each auxiliary order (e.g., creation order $\to_c$ and mutex orders $\to_m$) meets its requirements;
    \item The transitive closure of the combined program and auxiliary orders is acyclic,
        with a least element representing the main thread's initial configuration and a maximal element.
\end{itemize}
The maximal element is the ego thread's last configuration,
ensuring that the local trace represents one thread's perspective.
\cref{f:prog1} shows an example representation for a local trace of the program in \cref{f:barrier} with ego thread $w1$.

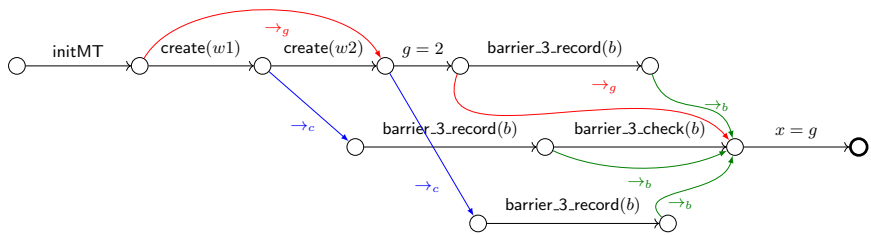
\begin{figure*}[t]
    \ltDrawingDefs
    \centering
    \scalebox{0.72}{
        \begin{tikzpicture}
            \tikzset{programpointbarrier/.style={programpoint,node distance=40pt and 90pt}}
            \node[programpoint](mainpp0){};
            \programon[wide]{mainpp0}{\initMT}{mainpp1}
            \programon[wide]{mainpp1}{\create(w1)}{mainpp2}
            \programon[wide]{mainpp2}{\create(w2)}{mainpp3}
            \programon{mainpp3}{g = 2}{mainpp4}
            \programon[barrier]{mainpp4}{\textsf{barrier\_3\_record}(b)}{mainpp5}

            \node[programpoint,below right=36pt and 42pt of mainpp2](w1pp0){};
            \programon[barrier]{w1pp0}{\textsf{barrier\_3\_record}(b)}{w1pp1}
            \programon[barrier]{w1pp1}{\textsf{barrier\_3\_check}(b)}{w1pp2}
            \programon[widesink]{w1pp2}{x = g}{w1pp3}

            \node[programpoint,below right=76pt and 42pt of mainpp3](w2pp0){};
            \programon[barrier]{w2pp0}{\textsf{barrier\_3\_record}(b)}{w2pp1}

            \createo[left=2mm, pos=0.8]{mainpp2}{w1pp0}
            \createo[left=2mm, pos=0.8]{mainpp3}{w2pp0}
            \signalo[right=1mm, pos=0.5]{mainpp5}{w1pp2}{b}{out=285,in=105,looseness=1.25}
            \signalo[below=1mm, pos=0.5]{w1pp1}{w1pp2}{b}{out=335,in=205,looseness=0.8}
            \signalo[right=1mm, pos=0.18]{w2pp1}{w1pp2}{b}{out=135,in=255,looseness=1.15}
            \mutexo[below=1.5mm, pos=0.55]{mainpp1}{mainpp3}{g}{out=60,in=120,looseness=0.8}
            \mutexo[above=2mm, pos=0.62]{mainpp4}{w1pp2}{g}{out=250,in=130,looseness=0.9}
        \end{tikzpicture}}
    \caption{Stylized local trace for \cref{f:barrier} with ego thread $w1$,
    assuming the first step of $\textsf{compute}(g,arg)$ is storing $g$ into $x$.
    Black arrows indicate the order of configurations within threads,
    red arrows indicate the order induced by accesses to $g$.
    Green arrows indicate the barrier order from the recorded arrivals at $b$ to the check, %
    blue arrows the order between a configuration creating a new thread and its first
    configuration.
    }\label{f:prog1}
\end{figure*}

Let $\Traces$ denote the program's local traces, defined inductively later.
For now, let $\init \subset \Traces$ contain the \emph{initial} local traces, each consisting only of the main thread's initial configuration.
Given a local trace $t$ and a CFG edge $e = (u,\act,v)$ where $\act$ is not observing,
we define the semantics by a function $\semT{e}(t): \Traces \to (\Traces \option)$ which either returns
$\oSome t'$ if $t$ can be extended to a new local trace $t'$ by executing $\act$ along $e$,
or $\oNone$ otherwise, e.g., if $t$ ends in a program point $u' \neq u$.
For $k$-ary observing actions $\act$, we define the semantics using a function
$\semT{e}(t_0,\ldots,t_{k-1}): \Traces^k \to (\Traces \option)$ which returns the resulting local trace
if the provided local traces can be combined into a new local trace.\footnote{
Thus, successor computation is deterministic; nondeterminism can instead be encoded by multiple CFG edges.
We use the curried version of $\semT{\cdot}$ interchangeably.
    }
If $\semT{e}(t_0,\ldots,t_{k-1}) = \oSome t'$, then $t_0,\ldots,t_{k-1}$ are \emph{compatible} w.r.t. the edge $e$ and each other.
Otherwise, they are \emph{incompatible}.

For example, $\semT{(u,\lock(a),v)} \, t_0 \, t_1$ returns $\oSome t'$ if
$t_0$ ends in $u$ and $t_1$ ends in an \texttt{unlock(a)}
action \emph{and} taking the union of these two local traces and adding an edge $\to_a$ from \texttt{unlock(a)}
to the newly introduced successor node of \texttt{lock(a)} yields a new local trace $t'$. Otherwise, it returns $\oNone$.
There is also a function $\new: \Nodes \to \Nodes \to \Traces \to (\Traces\option)$, where
$\new\,u\,u_1\,t = \oSome t'$ computes the initial local trace of the new thread, if there is an
edge $(u,\create\,u_1,u')$ in the CFG and $t$ ends in $u_1$, and $\oNone$ otherwise.
The set $\Traces$ is the least fixpoint obtained by applying $\semT{\cdot}$
and $\new$ to $\init$ and to traces generated along the way.
For abstract interpretations, an equivalent constraint-system view is more useful.
We use variables (\emph{unknowns}) $\relax[u]$ for each program point $u\in\Nodes$ and $\relax[\act]$ for observable actions $\act$.
The resulting constraint system then takes the following form:
\begin{equation}
    \begin{array}{lllr}
        [u_0] &\supseteq& \init& \\[0.3em]
        \relax[\act] &\supseteq& \semT{(u_1,\act,u_2)}(\relax[u_1])&
         (\text{for } (u_1,\act,u_2) \in\Edges, \act \text{ observable})\\[0.3em]
        \relax[u_2] &\supseteq& \semT{(u_1,\act,u_2)}(\relax[u_1]) \qquad\qquad&
         (\text{for } (u_1,\act,u_2) \in\Edges, \act \text{ not observing})\\[0.5em]
        \relax[u_2] &\supseteq& \semT{(u_1,\act,u_2)}(\relax[u_1],\relax[\act_0], \dots,\relax[\act_{k-2}])\span\\
         \span\span\span (\text{for } (u_1,\act,u_2) \in\Edges, \act \text{ of arity $k$},
          \act_0 \dots \act_{k-2} \text{ observed by } \act)\\[1em]
         \relax[u'] &\supseteq& \new(\relax[u_1],u_1,u')&
            (\text{for } (u_1,x = \create\,u',u_2) \in\Edges)
    \end{array}
\label{eq:local-trace-constraints}
\end{equation}
where $u_0$ is the main thread's initial program point and the functions $\semT{\cdot}$
and \new\ are lifted point-wise to sets of local traces.
For an edge $e$ with an action of arity $k$ and sets $T_i, 0 \leq i < k$ of local traces we define
$
\semT{e}(T_0,\ldots,T_{k-1}) \coloneqq \bigcup_{t_0\in T_0,\ldots,t_{k-1}\in T_{k-1}}
\begin{cases}
    \{t'\} & \text{if } \semT{e}(t_0,\ldots,t_{k-1}) = \oSome t'\\
    \emptyset & \text{otherwise}
\end{cases}.
$\\
In \cref{eq:local-trace-constraints}, observable-action edges have two constraints.
One propagates traces to the CFG successor unknown; the other records them in the observable action unknown.
In the concrete semantics, the right-hand sides for both are identical.
We leave local actions unspecified as our approach is parametric in them.

\paragraph{Data Races.}
Since local trace semantics totally orders all accesses to a global $g$ via $\to_g$ edges, its use for data-race detection may seem surprising.
\citet{SchwarzE26} show that it is possible: \emph{bidirectional trace compatibility} coincides with the intuitive notion
of two accesses to a global being unordered.
\begin{definition}
    Traces $t_0$ and $t_1$
    ending in nodes $u_0$ and $u_1$ with outgoing edges $(u_0,\act_0,u_0')$ and $(u_1,\act_1,u_1')$
    where $\act_0$ and $\act_1$ are accesses to the same global variable,
    are \emph{bidirectionally trace compatible}, if, with appropriately lifted $\sem{e}$,
    \[
    \begin{array}{lll}
        \exists t' \in \Traces&:& \oNone \neq \semT{(u_0,\act_0,u_0')}(t_0,\semT{(u_1,\act_1,u_1')}(t_1,t'))\land\\
        && \oNone \neq \semT{(u_1,\act_1,u_1')}(t_1,\semT{(u_0,\act_0,u_0')}(t_0,t')).
    \end{array}
    \]
\end{definition}
This means that the accesses can immediately follow each other in either order,
matching the intuitive notion of being unordered.
\begin{proposition}\label{prop:data-race-bidirectional-compat}
    Two accesses to a global variable $g$ are racy in the intuitive sense if and only if
    at least one is a write and
    traces $t_0$ and $t_1$ ending in the respective CFG predecessors exist, s.t.\ they
        are bidirectionally trace compatible.
\end{proposition}
\begin{proof}
    Similar to \cite[Theorem 1]{SchwarzE26}; see
    \apprefInlineShort{G}{app:prop1-prop2}.\qed
\end{proof}

\subsection{Digests and Digest-Driven Data Race Detection}\label{ss:digests}

\emph{Digests}~\cite{schwarzthesis,hammerandnail,schwarz2023clustered} abstract thread histories
and can refine abstract interpretations of concurrent programs.
They exploit that some thread histories are \emph{incompatible}, a notion transferable from local traces to digests.

Let $\Digests$ be the set of digests.
$\alpha_\Digests$ maps a local trace $t$ to its digest $\alpha_\Digests\, t \in \Digests$.
\footnote{We denote by $\gamma_\Digests\,A = \{ t \mid \alpha_\Digests\,t = A \}$ the concretization function
for a digest.}
We require digests to be computable inductively.
For the initial set of digests at the program start, we define
\begin{equation}
\initDigests = \{ \alpha_\Digests\,t \mid t \in \init\}.
\label{def:initSound}
\end{equation}
Further, we introduce a digest counterpart $\semDigests{u,\act} : \Digests^k \to (\Digests \option)$ of the semantics $\semT{u,\act,v}$
inheriting its arity from the concrete semantics but operating on digests instead of traces. We call these functions
right-hand side functions.
More concretely, for an edge $e = (u,\act,v)$ we demand $\forall t_0,\cdots,t_{k-1} \in \Traces$
\begin{equation}
    \begin{array}{lll}
        \textsf{map}\;\alpha_\Digests\;(\semT{e}(t_0,\cdots,t_{k-1}))
        \hat{\sqsubseteq} \semDigests{u,\act}(\alpha_\Digests\,t_0,\cdots,\alpha_\Digests\,t_{k-1})
    \end{array}
    \label{def:SemSound}
\end{equation}
where $\alpha_\Digests$ is mapped over the trace option returned by $\semT{e}(t_0,\cdots,t_{k-1})$, and $\hat{\sqsubseteq}$
is the partial order on $\Digests \option$ with $\oNone$ as bottom and all other elements incomparable.
\Cref{def:SemSound} says that $\aSemDigest{\cdot}$ soundly overapproximates $\semT{\cdot}$.

For newly created threads, a function $\newDigests : \Digests \to \Nodes \to \Nodes \to \Digests \option$
returns the digest of a thread starting at $u_1$, created at $u$ by a thread with digest $A$.
We demand that $\newDigests$ soundly abstracts $\new$:
\begin{equation}
    \forall t \in \Traces:
    \textsf{map}\;\alpha_\Digests\;(\new\;t\;u\;u_1) \hat{\sqsubseteq}
    \newDigests\;(\alpha_\Digests\;t)\;u\;u_1
    \label{def:NewSound}
\end{equation}
If the thread calling $\create$ has a successor digest, then $\newDigests$ must return a digest for the created thread:
\begin{equation}
    \forall A \in \Digests:
    \aSemDigest{u, x = \create(u_1)}(A) \neq \oNone \Longrightarrow
    \newDigests\;A\;u\;u_1 \neq \oNone
    \label{def:CreateNewConsistent}
\end{equation}
Digests satisfying \labelcref{def:initSound,def:SemSound,def:NewSound,def:CreateNewConsistent} are \emph{admissible}.

\paragraph{Digest-Driven Abstract Interpretation.}
An analysis can be refined using any admissible digest.
We quickly present the requirements for an analysis.
An analysis constraint system is obtained from the concrete constraint system by replacing sets of traces with an analysis domain
and concrete transfer functions with abstract ones.
Let $\D$ be the abstract domain used by the analysis, which forms a complete lattice with operations join $\sqcup$ and meet $\sqcap$.
A monotonic function $\gamma_\D: \D \to 2^{\Traces}$ provides the concretization of a given domain element.
The analysis must define abstract transfer functions $\aSem{\cdot}$, as well as abstract versions $\newAbs$ and $\initAbs$ of $\new$ and $\init$.
Given soundness of $\aSem{\cdot}$, $\newAbs$, and $\initAbs$ w.r.t.\ their concrete counterparts, related by $\gamma_\D$,
and the monotonicity of $\aSem{\cdot}$ and $\sem{\cdot}$,
it follows that the concretization of a solution of the analysis constraint system is greater than the least
solution of the concrete constraint system (see \cite[Theorem~1]{hammerandnail}).

When refining an analysis with a digest, the constraint system is adapted by splitting unknowns according to digests,
yielding a constraint system of form:
\[
    \begin{array}{lllr}
        [u_0, A_0] & \sqsupseteq& \initAbs& (\text{for } A_0 \in \initDigests) \\[0.5em]
        \relax[\act, A'] &\sqsupseteq& \aSem{(u_1,\act)}(\relax[u_1, A])&
         (\text{for } (u_1,\act,u_2) \in\Edges, \act \text{ observable}, \\
         & & &  \oSome A' = \aSemDigest{u_1,\act}\ A)\\[0.5em]
        \relax[u_2, A'] &\sqsupseteq& \aSem{(u_1,\act,u_2)}(\relax[u_1, A])& (\text{for } (u_1,\act,u_2) \in\Edges, \act \text{ not observing}, \\
         & & &  \oSome A' = \aSemDigest{u_1,\act}\ A)\\
                 \relax[u_2, A'] &\sqsupseteq& \aSem{(u_1,\act,u_2)}(\relax[u_1, A_1],\relax[\act_0, A_2], \dots,\relax[\act_{k-2}, A_k])\span\\
         \span\span\span (\text{for } (u_1,\act,u_2) \in\Edges, \act \text{ of arity $k$},
          \act_0 \dots \act_{k-2} \text{ observed by } \act, \\
          & & &  \oSome A' = \aSemDigest{u_1,\act}\ (A_1, A_2, \dots A_k))\\[1em]
         \relax[u', A'] &\sqsupseteq& \newAbs(\relax[u_1, A],u_1,u')\qquad\qquad&
            (\text{for } (u_1,x = \create\,u',u_2) \in\Edges, \\
            & & & \oSome A' = \newDigests A\;u_1\;u')
    \end{array}
\]

\begin{proposition}\label{prop:refine-sound}
    Refining a sound abstract interpretation of a concurrent program with an admissible digest yields
    a sound abstract interpretation.
\end{proposition}
\begin{proof}
    By giving a digest-refined concrete semantics and showing that soundness of the original abstract interpretation w.r.t.\ the original
    semantics carries over to the refined setting. %
    See \apprefInline{G}{app:prop1-prop2} for details.\qed
\end{proof}

As examples, we consider the lockset digest, which tracks each thread's currently held
locks (see \cref{f:locksets}), and the thread \emph{id} digest for dynamically created threads:
\begin{figure*}[t]
    \begin{minipage}[t]{.38\linewidth}\[
        \begin{array}{lll}
            \Digests = 2^\Mutexes\\
            \initDigests = \{\emptyset\}\\
            \newDigests\,S\,u_1 = \oSome \emptyset\\
            \semDigests{u,\act}\,(S_0,\cdots,S_i) = \oSome S_0\\
            \qquad\text{(other edge } (u,\act,v) \text{)}
        \end{array}
    \]
    \end{minipage}
    \begin{minipage}[t]{.57\linewidth}\[
        \begin{array}{rll}
            \semDigests{\_,\lock(a)}\,(S_0,S_1) &=& \begin{cases}
                \oNone & \text{if } a \in S_0\\
                \oSome (S_0 \cup \{a\}) & \text{otherwise}
                \end{cases}\\[1ex]
            \semDigests{\_,\unlock(a)}\,S &=&
                \begin{cases}
                    \oNone & \text{if } a \not\in S\\
                    \oSome (S \setminus \{a\}) & \text{otherwise}
                \end{cases}
        \end{array}
        \]
    \end{minipage}
    \caption{Right-hand sides for lockset digest~\cite{hammerandnail}, generalized to our setting.
     Underscores serve as wildcards that do not bind the matched value.
    }\label{f:locksets}
\end{figure*}
\citet{schwarz2023clustered} propose a digest computing thread \emph{id}s for \textsc{pthreads}.
It identifies threads by the sequence of create edges involved in their creation:
the main thread receives \textsf{main}, a thread it creates along edge $e$ receives $\textsf{main} \cdot e$, and so on.
To account for cycles in the creation history caused by recursively
creating threads or threads being created in a loop, non-definite thread \emph{id}s are used. %
For details, see \cite{schwarz2023clustered}. %
In the rest of the paper, we use generic thread \emph{id}s with the following assumptions.
A thread \emph{id} is information that remains constant for all operations of a given
concrete thread and can be projected from a digest $\NDigests{1}$ by $\pTID: \NDigests{1} \to \TIDs$.
We further assume the existence of functions:
\begin{itemize}
    \item $\unique: \TIDs\to\Bool$ checking whether, in any execution, at most one concrete thread can have this abstract thread \emph{id}.
    \item $\mustanc: \TIDs \to 2^{\TIDs}$ returning the set of \emph{must ancestors}, i.e., threads that must have been involved in the creation of the current thread.
    \item $\maycreate: \NDigests{1} \to 2^{\TIDs}$ returning the \emph{may-created} thread \emph{id}s, i.e., child-thread \emph{id}s the ego thread may have created given a digest.
\end{itemize}
The product of admissible digests is also an admissible digest \cite{hammerandnail}. %
This allows for a modular design where digests tracking different aspects of the history can seamlessly be combined during an analysis.

\paragraph{May-Happen-In-Parallel Predicate.} Digests support reasoning about the temporal
separation aspect of data races~\cite{SchwarzE26}.
For two accesses to race, they must potentially happen concurrently: by \cref{prop:data-race-bidirectional-compat},
there must be bidirectionally trace-compatible traces ending just before them.
Thus, compatibility of \emph{digests}, an abstraction of local-trace compatibility, can exclude data races.
Two digests $A_0$ and $A_1$ are compatible with an edge $e$ and each other if there are traces $t_0$ and $t_1$
with these digests that are compatible with $e$ and each other.
Thus, if two digests are not bidirectionally compatible, the corresponding traces are not either.
This yields a first definition of a \emph{may-happen-in-parallel} predicate
\[
||_{e_1, e_2}^? : \Digests \to \Digests \to \{ \textsf{false}, \top \}
\]
for checking whether two accesses to a variable may race,
where $e_1$ and $e_2$ are the CFG edges in question. We outline the detailed requirements for
$||_{e_1, e_2}^?$ in \apprefInline{A}{app:mhp}.
When edges are clear or irrelevant, we write $||^?$.
The answers form the lattice $\textsf{false} \sqsubseteq \top$.
For product digests, answers can be combined using meet ($\sqcap$).
For thread \emph{id}s, the $||^?$ definition arising from bidirectional digest compatibility is maximally precise.
However, this is not the case for all digests.
For instance, with the lockset digest, two actions with overlapping locksets cannot run concurrently, so the corresponding concrete accesses cannot race.
However, bidirectional digest compatibility cannot infer this, because overlapping locksets need not prevent the respective accesses to a global $g$ from directly following each other in the $\to_g$ order.
Thus, one may instead provide a more precise $||^?$ and justify it directly w.r.t. bidirectional \emph{trace} compatibility.
We generalize $||^?_{e_1, e_2}$ to show that two arbitrary-action edges $e_1,e_2$ cannot happen in parallel:
introduce a fresh global variable and replace the actions by writes to it, yielding edges $e_1',e_2'$.
Then $A_1 ||^?_{e_1, e_2} A_2$ is defined as $A_1 ||^?_{e_1', e_2'} A_2$.

\paragraph{Race Detection.} \citet{SchwarzE26} give a data race detection algorithm using digests.
First, one performs digest-driven abstract interpretation, collecting digests associated with accesses.
Then, for each shared memory location, one iterates over all access pairs.
For accesses at edges $e_1$ and $e_2$ with digests $A_1$ and $A_2$,
a potential race is flagged if at least one access is a write and
$A_1 ||^?_{e_1, e_2} A_2 = \top$.

\section{Synchronization via \pthreadOnce}\label{ss:once}
\begin{figure}[t]
    \begin{minipage}[t]{7cm}
\begin{minted}[tabsize=1]{c}
// Library file
pthread_once_t once;
schmilblick* dev; // Some device

void init_lib() {
  dev = malloc(size_of(schmilblick));
  init_schmilblick(dev);
}

void maluba(){
  // Ensure the library is initialized
  pthread_once(once,init_lib);
  assert(dev != NULL);
  // ...
}
\end{minted}
\end{minipage}\begin{minipage}[t]{7cm}
    \begin{minted}[tabsize=1]{c}
// Main file
void *t1(void *arg) {
  maluba();
  // ...
  return NULL;
}

int main() {
  pthread_t t;
  pthread_create(&t, t1, NULL);

  maluba();
  // ...
}
    \end{minted}
        \end{minipage}
\caption{\textsc{Pthread} program using \pthreadOnce}\label{f:onceComplete}
\end{figure}
\pthreads\ offers a mechanism for executing code at most once.
This enables resource initialization before thread access, such as opening a shared log file or database connection.
Application developers can place such code in \texttt{main}, but library developers do not control \texttt{main}
and instead rely on \pthreadOnce~\cite{butenhof1993programming}.
It provides a type for control variables (which we assume are from set $\mathcal{O}$) and a function
\texttt{pthread\_once(o,f)} which executes $f$ and sets $o$ if it is unset, and otherwise does nothing.
The system makes running $f$ and setting $o$ appear atomic.
Consider the program in \cref{f:onceComplete}.
\texttt{init\_lib} must allocate and initialize the device before any other library function runs.
\pthreadOnce ensures that \texttt{init\_lib} is called exactly once and that the device is initialized before use.

\pthreadOnce is often overlooked by static analyzers or abstracted very crudely.
With digests, it can be handled precisely without complicating the underlying analysis;
all handling happens inside the digest.
This excludes two race classes: within the called function and between initialization and later code.

\paragraph{Extension of the concrete semantics.}

We introduce a new observable action $\onceInit(o)$ to initialize a control
variable $o$ at program start\footnote{
    Unlike library developers, a whole-program analyzer can conceptually insert helper code at program start,
    so this action does not negate the need for \pthreadOnce.
} and compile \texttt{pthread\_} \texttt{once(o,f)} to a sequence calling $\onceStart(o)$, branching
on $\onceRan(o)$, executing the inlined code from $f$ if $\onceRan(o)$ is false, and finally performing $\onceEnd(o)$.
Here, $\onceInit(o)$ and $\onceEnd(o)$ are observable, while $\onceStart(o)$ is observing,
analogous to $\initAct/\unlock(a)$ and $\lock(a)$ for mutex $a$. The once order $\to_o$ follows
\begin{wrapfigure}{r}{0.6\textwidth}
    \centering
    \begin{minipage}[t]{0.25\textwidth}
\begin{minted}[tabsize=1]{c}
main:
 initO(o);
 create(t1);
 // ...
 startO(o);
 if (!ran(o))
    dev = malloc(42);
 endO(o);

 assert(dev != NULL);

\end{minted}
    \end{minipage}\begin{minipage}[t]{0.25\textwidth}
\begin{minted}[tabsize=2]{c}
t1:
 // ...
 startO(o);
 if (!ran(o))
    dev = malloc(42);
 endO(o);

 // ...

 assert(dev != NULL);
\end{minted}
    \end{minipage}
\caption{Program highlighting how \pthreadOnce is decomposed into individual steps.}\label{f:once}
\vspace{-2em}
\end{wrapfigure}
the mutex-order requirements from~\cref{ss:lt}.
The actions $\Pos(\onceRan(o))$ and $\Neg(\onceRan(o))$ guard whether the code associated with
\pthreadOnce is to be executed: the
former returns a local trace only when $\onceEnd(o)$
appears in the trace, and the latter only when it does not.\footnote{
Recursively calling \texttt{pthread\_once(o,f)} is not conclusively specified in the standard; here,
we decided to go for the behavior of \texttt{glibc}, which deadlocks upon such calls.
}
The program in \cref{f:once} corresponds to the one in \cref{f:onceComplete} after abstracting away
details and expressing \textsf{pthread\_once} as above.

A digest for this feature (\cref{f:onceASimple}) may then track two sets in a tuple $(A,C)$:
\begin{itemize}
    \item A set $A$ of \emph{active} control variables, i.e., those for which the ego thread is currently inside the section
    between $\onceStart(o)$ and $\onceEnd(o)$.
    \item A set $C$ of \emph{completed} control variables, i.e., those for which the associated fragment
    has surely been run, as a call to $\onceEnd(o)$ was observed.
\end{itemize}
The guards $\Pos(\onceRan(o))$ and $\Neg(\onceRan(o))$ check whether $o$ has already completed, i.e., whether it is in $C$.
$\onceStart(o)$ adds $o$ to the active set, and $\onceEnd(o)$ moves it from the active to the completed set.
\begin{figure*}[b]
    \begin{minipage}[t]{.4\linewidth}\[
        \begin{array}{lll}
        \Digests = 2^\Onces\times 2^\Onces\\[0.5ex]
		\initDigests = \{(\emptyset,\emptyset)\}\\[0.5ex]
        \newDigests\,(A,C)\,u\,u_1= \oSome (\emptyset,C)\\[0.5ex]
        \semDigests{\_, \Pos(\onceRan(o))}\,(A,C) =\\
        \qquad \begin{cases}
            \oNone & \text{if } \texttt{o} \not\in C\\
            \oSome (A,C) & \text{otherwise}
        \end{cases}\\
        \semDigests{u,\act}\,(A,C)\,\cdots\,(A_i,C_i) =\\
          \qquad \oSome (A,C) \quad \text{(other edge } (u,\act,v) \text{)}\qquad\;\;\\
        \end{array}%
    \]
    \end{minipage}%
    \begin{minipage}[t]{.5\linewidth}\[
        \begin{array}{lll}
            \semDigests{\_, \onceStart(o)}\,(A_0,C_0)\,(A_1,C_1) =\\
            \qquad \begin{cases}
                \oSome (A_0 \cup \{o\}, C_0 \cup C_1) & o \not\in A_0\\
                \oNone & \text{otherwise}
            \end{cases}\\
            \semDigests{\_, \onceEnd(o)}\,(A,C) = \oSome (A \setminus \{o\}, C \cup \{o\})\\[0.5ex]
            \semDigests{\_, \Neg(\onceRan(o))}\,(A,C) =\\
            \qquad \begin{cases}
                \oNone & \text{if } \texttt{o} \in C\\
                \oSome (A,C) & \text{otherwise}
            \end{cases}
        \end{array}
        \]
    \end{minipage}
    \caption{Right-hand sides for excluding races based on \pthreadOnce.}\label{f:onceASimple}
\end{figure*}
$||^{?}$ can then be defined as:
\[
(A,C)\;||^{?}\;(A',C') = \begin{cases}
    \textsf{false} & \text{ if } (A \cap (A' \cup C') \neq \emptyset ) \lor (A' \cap (A \cup C) \neq \emptyset)\\
    \top & \text{otherwise},
\end{cases}
\]
i.e., accesses may not race when a control variable is active during both accesses, or active in one and completed in the other.
\begin{example}
    Consider again the program in \cref{f:once}. An analysis using the digest from \cref{f:onceASimple}
    succeeds in establishing race freedom.
    $\hfill\qed$
\end{example}
\begin{proposition}\label{p:onceSound}
  The digest given above is admissible, and its predicate $||^?$ sound.
\end{proposition}
\noindent
See \appref{C}{app:oncelong} for the proof (by contradiction) and
an extended digest making the analysis of global variables more precise.

\section{Synchronization via Barriers}\label{ss:barriers}

\pthreads\ provides barriers: after initialization by calling
$\textsf{pthread\_barrier\_init}(b,c)$ for a barrier $b$ and
capacity $c\in\mathbb{N}^+$, callers to $\textsf{pthread\_barrier\_wait}(b)$ block until $c$ threads
have called it, at which point all waiters are unblocked.
\Cref{f:barrier} shows a barrier example whose $\textsf{pthread\_barrier\_wait}(b)$ calls may all happen in parallel.

In fact, this always needs to be the case, as any call to $\textsf{pthread\_barrier\_wait}(b)$ happening \emph{after}
another call has returned, cannot contribute to reaching the required
capacity.
The predicate $||^?$ also detects stuck executions when fewer than $c-1$ other
$\textsf{pthread\_barrier\_wait}(b)$ calls may happen in parallel.
Passing a barrier reveals cross-thread information that excludes some races.

\paragraph{Extension of the concrete semantics.}
Let $\Barriers$ be a set of barriers.
For presentation, we initially assume static capacities (eliminating \textsf{pthread\_barrier\_init})
and replace each $\textsf{pthread\_barrier\_wait}(b)$ with
$\textsf{barrier\_N\_record}(b);$ $\textsf{barrier\_N\_check}(b)$, where $N \in \mathbb{N}^+$ is $b$'s capacity.
The action $\textsf{barrier\_N\_record}(b)$
is non-blocking and observable, i.e., the local trace ending in it is recorded at the corresponding unknown.
The action $\textsf{barrier\_N\_check}(b)$ is observing and $N$-ary, incorporating local traces
stored at the unknown corresponding to $\textsf{barrier\_N\_record}(b)$, i.e., those ending in
	that action.\footnote{This uses our generalization to higher-arity observing actions.}
\Cref{f:prog1} shows an example local trace. %

By construction of the constraint system, for $\semT{e}(t_0, \dots, t_{N-1})$ for an edge $e=(u,\textsf{barrier\_N\_check}(b),v)$,
$t_0$ is a local trace ending in $u$ with last action
$\textsf{barrier\_N\_record}(b)$, and the other $N-1$ local traces provided as an argument also end
in $\textsf{barrier\_N\_record}(b)$, as that is the action observed by the check action.
The call yields $\oSome t'$ only if its arguments are distinct and their union,
augmented with dependency edges from the record endpoints to the check endpoint, is a local trace.
For each barrier $b$, the barrier order $\to_b$
goes from nodes with an incoming edge
$\textsf{barrier\_N\_record}(b)$ to those with an incoming edge $\textsf{barrier\_N\_check}(b)$;
call these sets $B^+$ and $B^-$, respectively.
Each $B^+$ node has at most $1$ outgoing $\to_b$ edge,
while $B^-$ node has \emph{exactly} $N$ incoming $\to_b$ edges,
including one from the control-flow predecessor, i.e., the thread passing the barrier was involved in
the barrier being passed. This order then forms part of the causality order of a local trace.
Consider the $N$ subtraces ending at $B^+$ nodes
that directly $\to_b$-precede a $B^-$ node.
None of these subtraces may contain any of the relevant calls to $\textsf{barrier\_N\_record}(b)$.
This ensures that all $N$ calls happened in parallel.

Before deriving a digest tracking information about barriers, we first observe that for any digest
that provides the predicate $||^?$ for checking for data races and does not track information about barriers, i.e., is
unaffected by executing actions $\textsf{barrier\_N\_record}(b)$ and $\textsf{barrier\_N\_check}(b)$,
a definition of $\aSemDigest{\_, \textsf{barrier\_N\_check(b)}}$ can be given that detects some instances where the execution
must become stuck: $\aSemDigest{\_, \textsf{barrier\_N\_check(b)}}(A_0,\dots,A_{N-1}) =$
\[
    \begin{array}{l}
         \begin{cases}
            \oSome A_0 & \text{if } \forall i,j \in [0,N-1]: i \neq j \implies A_i ||^{?} A_j = \top\\
            \oNone & \text{otherwise}
        \end{cases}
    \end{array}
\]
This uses the predicate $||^?$ to check whether the $N$ calls to $\textsf{barrier\_N\_record}(b)$ can all happen in parallel.
When this is not the case, the execution must become stuck, and $\oNone$ is returned, denoting unreachability.
The predicate thus turns out to be not only useful when considering whether accesses can happen in parallel, but also when considering richer synchronization primitives.
\begin{example}
    Consider \cref{f:barrier} and assume that the developer has, by mistake, surrounded calls to
    $\textsf{pthread\_barrier\_wait}(b)$ with locking and unlocking of mutex $m$.
    Then, one thread reaches the call to $\textsf{pthread\_barrier\_wait}(b)$ while holding $m$.
    However, it will not release $m$ until the barrier is passed, and the other threads do not
    call $\textsf{pthread\_barrier\_wait}(b)$
    until they have acquired $m$. The execution deadlocks.
    With the lockset digest from \cite{SchwarzE26} with the definition of $||^?$
    returning false when locksets overlap, the analysis will detect this as
    $\{m\} ||^{?} \{m\} = \textsf{false}$.
    $\hfill\qed$
\end{example}

Showing certain accesses unreachable may improve the precision of race detection.
If a digest tracks thread \emph{id}s, we can define a more precise digest that excludes more races
by tracking barriers passed by the current thread and observed barrier actions of other threads.
If, for two accesses, one thread knows that the other has passed a barrier that it has not yet passed,
the accesses are ordered w.r.t.\ each other and thus not bidirectionally compatible and not racy.

The new digest is the product of the original digest $\NDigests{1}$ and two components%
\footnote{The tracking of separate components for $R$ and $O$ is to simplify presentation here: Whenever the thread
\emph{id} in $A$ is unique, $R$ can be reconstructed from $O$.}:
\begin{itemize}
    \item a set $R \subseteq \Barriers$ of barriers for which the current thread has called $\textsf{barrier\_N\_record}(b)$
    \item a map $O: \TIDs \to 2^\Barriers$ from thread \emph{id}s to sets of barriers for which the respective thread has been
    \emph{observed} to have called $\textsf{barrier\_N\_record}(b)$.
\end{itemize}
The corresponding right-hand sides for the digest are given in \cref{f:barrierA} in terms of those for $\NDigests{1}$.
The predicate then is given by: $(A_0,R_0,O_0)\;||^{?}\;(A_1,R_1,O_1) = $
\[
    \begin{array}{l}
        \qquad = (A_0\;||^{?}_{\NDigests{1}}\; A_1) \sqcap
        \begin{cases}
            \textsf{false} & \text{if }  (\unique\,(\pTID\,A_0) \land O_1\,(\pTID\,A_0) \not\subseteq R_0)\\
            &\quad     \lor (\unique\,(\pTID\,A_1) \land O_0\,(\pTID\,A_1) \not\subseteq R_1)
            \\
            \top & \text{otherwise}
        \end{cases}
    \end{array}
\]
On top of the predicate for the base digests, this exploits the fact that if $A_0$ is unique, any of its actions that have been observed by the other thread must already have happened for the accesses to potentially happen concurrently.

\begin{figure*}[t]
\[
\begin{array}{l}
    \Digests = \Digests_1 \times 2^\Barriers \times (\TIDs \to 2^{\Barriers} )\qquad\qquad\qquad
    \initDigests = \{(A_1,\emptyset,\emptyset) \mid A_1 \in \initNDigests{1} \}\\[1ex]
    \newDigests\,(A,R,O)\,u\,u_1 =
        \begin{cases}
            \oSome (A',\emptyset,O) & \text{if } \newNDigests{1}\,   A\,u\,u_1 = \oSome A'\\
            \oNone & \text{otherwise}
        \end{cases}\\[1ex]

    \aSemDigest{\_, \textsf{barrier\_N\_check(b)}}((A_0,R_0,O_0),\dots,(A_{N-1},R_{N-1},O_{N-1})) =\\
    \quad \begin{cases}
        \oSome (A_0, R_0, \bigcup_{i} O_i) &
        \text{if } \forall i,j \in [0,N-1]: (i \neq j \implies A_i ||^{?} A_j) \land\\
        & \qquad\qquad(\unique\,(\pTID\,A_i) \implies O_j\,(\pTID\,A_i) \subseteq R_i) \\[1ex]
        \oNone & \text{otherwise}
    \end{cases}\\[1ex]

    \aSemDigest{\_, \textsf{barrier\_N\_record(b)}}(A,R,O) =\\
        \qquad\oSome
        (A, R \cup \{ b \}, O \oplus \{ (\pTID\,A) \mapsto (O\,(\pTID\,A)) \cup \{b\} \})\\[1ex]

    \aSemDigest{\_, \act}((A_0,R_0,O_0),\dots,(A_{N-1},R_{N-1},O_{N-1})) =\\
        \qquad \begin{cases}
        \oSome (A', R_0, \bigcup_i O_i) & \text{if } \aSemNDigest{1}{\_,act}(A_0,\dots,A_1) = \oSome A'\\
        \oNone & \text{otherwise} \qquad\qquad\qquad\qquad\qquad\qquad \text{(other $n$-ary action)}
        \end{cases}
\end{array}
\]\vspace{-1em}
\caption{Barrier digest where $\oplus$ denotes updating a map,
missing bindings are assumed to be $\emptyset$, and the union of maps is defined pointwise per binding.}\label{f:barrierA}
\end{figure*}

\begin{example}
    Consider again \cref{f:barrier} and assume a digest identifies all threads
    as unique\footnote{This is, e.g., the case for the digest-based thread \emph{id}s abstracting creation history~\cite{schwarz2023clustered}.}
    and assigns them \emph{id}s $\textsf{main}$, $\textsf{w1}$, and $\textsf{w2}$.
    The access in $\textsf{main}$ happens before calling $\textsf{pthread\_barrier\_wait}(b)$,
    and has digest $(A_0, \emptyset, \{ \})$; the worker accesses
    have digests $(A_j, \{ b \}, \{ \textsf{main} \mapsto \{ b\},
    \textsf{w1}\mapsto \{ b\},
    \textsf{w2}\mapsto \{ b\}
    \})$ for some $A_j$, where we omit empty bindings.
    As $\pTID\,A_0 = \textsf{main}$, $\unique\,\textsf{main}$ holds but $\{ b \} \not\subseteq \emptyset$,
    the write in $\textsf{main}$ happens before any of the reads,
    and race-freedom is shown.
    $\hfill\qed$
\end{example}
\begin{proposition}\label{prop:barr-sound}
    The digest given above is admissible, and its predicate $||^?$ sound.
\end{proposition}
The proof is by contradiction; see \appref{D}{sec:barriers-app}.
This digest generalizes straightforwardly to tracking how often barriers have been waited for, possibly stopping exact bookkeeping at a fixed constant $k$ to avoid a blowup.
Conversely, for a more tractable analysis, one can switch to an abstract digest to avoid keeping track of exact sets, along the lines of \cite{hammerandnail}.
\begin{remark}
    Our implementation supports dynamic barrier capacities and pointer-based accesses.
    Barriers may also be used for inter-process communication, but since the \GoblintAnon framework does not support process-level
    concurrency, our implementation soundly and coarsely treats such barriers as no-ops.
\end{remark}

\section{Synchronization via Mutexes and Creation Orders} \label{s:descendant-lockset}
While previous sections covered features that are often completely unhandled,
here we deal with complex patterns combining thread creation and mutexes:
Creating threads while holding a mutex can simulate barriers on platforms lacking them.
Also, then, worker threads do not have to wait for one another:
\begin{figure}[t]
  \centering
  \footnotesize
  \begin{minipage}[t]{.37\linewidth}
    \begin{minted}[tabsize=1,fontsize=\scriptsize]{c}
int g = 0;
pthread_mutex_t m;

void *worker(void *arg) {
  // Worker Initialization...

  pthread_mutex_lock(&m);
  pthread_mutex_unlock(&m);
  int f = compute(g, arg);
  // ...
  return NULL;
}
    \end{minted}
  \end{minipage}\hfill
  \begin{minipage}[t]{.58\linewidth}
    \begin{minted}[tabsize=1,fontsize=\scriptsize]{c}
int main() {
  // Worker threads
  pthread_t i1, i2;

  pthread_mutex_lock(&m);
  pthread_create(&i1, NULL, worker, NULL);
  pthread_create(&i2, NULL, worker, NULL);
  g++; // Initialization code

  pthread_mutex_unlock(&m);
  return 0;
}
    \end{minted}
  \end{minipage}
  \caption{Access to \texttt{g} in \texttt{main} must happen before accesses in worker threads.}
  \label{f:descendant-ls-example}
\end{figure}
\begin{example}
  \label{e:descendant-ls}
  In \cref{f:descendant-ls-example}, main creates workers $i_1$ and $i_2$ while holding $m$.
  Each worker initializes, then immediately acquires and releases $m$.
  These locks occur only after main unlocks $m$, so worker accesses to \texttt{g} follow main's access and cannot race.
  Notably, $i_1$ need not wait for $i_2$ to reach the lock.
\end{example}
For accesses in thread $i_0$ and a thread $i_1$ it created, the following conditions rule out conflicts via mutex $m$.
At the first access, we require that
\begin{enumerate}
  \item $i_1$ was created in $i_0$, ensuring $i_0$ is its ancestor.
  This check matters because an ancestor of $i_0$ might otherwise have created $i_1$, enabling races.
  \item At all of those creations, $m$ is locked---c.f.\ lockset digest (\cref{f:locksets}).
  \item $m$ was not unlocked after any such creation.
\end{enumerate}
\noindent Note that $i_1$ does not need to be a child of $i_0$, but can be more generally a descendant.
So, we account for \emph{transitive} creations of $i_1$, enhancing the precision.

At the second access, we record which threads have locked each mutex.
We can exclude a race if a lock of $m$ has happened at some point after the creation in $i_0$ in the
created thread or---for the same reason as above---one of its descendants.
If $m$ was locked in $i_H$, this is the case if $i_0$ must be an ancestor of
$i_H$. We describe a suitable digest in \apprefInline{E}{app:descendant-lockset-long}.

\section{Digests for Creation Locksets}\label{s:creation-lockset}

The pattern described in \cref{s:descendant-lockset} %
 is not the only one where thread creations with mutexes held can ensure mutual exclusion.
This time, we are interested in programs where a mutex $m$ is held in $i_0$ throughout the entire execution of a descendant thread $i_1$.
In consequence, $i_1$ is protected in some way by $m$.
This may sometimes be useful: Imagine, e.g., that a thread needs to do two expensive
computations in parallel while having exclusive access to a global. Then, it may spawn a new thread while holding the
protecting mutex to do one part of the work, while performing some other work in another task.
Before releasing the mutex again, the worker thread is joined again.
While this is probably more common for programming languages such as \texttt{Go}, where thread creation is cheap, it applies just as well to
C programs, as we show in the next example:

\begin{example}
  Consider the program given in \cref{f:creation-ls-example}.
  Throughout the entire execution of $work$, the mutex $m$ is locked.
  As $m$ is also in the lockset at the access to \texttt{g} in $t1$,
  this access cannot race with the one in $work$.
\end{example}

\begin{figure}
  \begin{minipage}[t]{5cm}
    \begin{minted}[tabsize=1]{c}
int g = 0;
pthread_mutex_t m ;
pthread_t i1, i2;

void* t1(void* arg) {
  pthread_mutex_lock(&m);
  g++; // not racing!
  pthread_mutex_unlock(&m);
  return NULL;
}

// main holds m throughout work
void* work(void* arg) {
  // Expensive computations
  g = ...; // not racing!
  return NULL;
}
    \end{minted}
  \end{minipage}
  \begin{minipage}[t]{7cm}
    \begin{minted}[tabsize=1]{c}
int main() {
  pthread_create(&i1, NULL, t1, NULL);
  ...
  pthread_mutex_lock(&m);

  // Two expensive computations
  // some of it happens in work

  // work created while m is locked
  pthread_create(&i2, NULL, work, NULL);
  // no unlock of m while work running
  pthread_join(i2, NULL);

  pthread_mutex_unlock(&m);
  return 0;
}
    \end{minted}
  \end{minipage}
  \caption{Creation Locksets: Mutex \texttt{m} protects accesses to \texttt{global}---even though
  $work$ never acquires it.}
  \label{f:creation-ls-example}
\end{figure}

Note that it is important to record that a mutex $m$ is locked in the creator thread $i_0$ and held throughout
the execution of a descendant thread $i_1$,
as it does not protect $i_1$ from accesses happening in $i_0$.
To detect that $i_1$ is protected from $i_0$ with $m$, we propose the following conditions:
\begin{enumerate}
  \item $i_1$ was transitively created in $i_0$, ensuring $i_0$ \emph{must} be an ancestor of it. We require the latter condition for the same reason as given in \cref{s:descendant-lockset}---there should be no alternative way of creating $i_0$ other than via $i_1$.
  \item At all of those creations, $m$ must be locked, as recorded by the lockset digest.
  \item There exists no unlock of $m$ while $i_1$ may run in parallel.
\end{enumerate}

For the last point, one can compute the \emph{id}s of the descendant threads running in parallel by determining all threads that were transitively created from the current thread and subtracting threads that were joined back. Using the thread \emph{id} digest, we obtain the former. For the latter information, we use the join digest suggested by \citet{hammerandnail}.
\appref{F}{app:creation-lockset-long} gives a digest for checking the conditions above.

\section{Experiments and Benchmark Set}\label{s:experiments}

\begin{table*}[t]
\caption{Runtime comparison. Runtime in seconds.}
\centering
\small
\begin{tabular}{@{}lrrrrr@{}}
\toprule
Benchmark &
\rotatebox{0}{LOC} &
\rotatebox{0}{Baseline} &
\rotatebox{0}{Once ($\Delta$\%)} &
\rotatebox{0}{Barriers ($\Delta$\%)} &
\rotatebox{0}{Locked ($\Delta$\%)} \\
\midrule
\texttt{cksum} & 68k & 1.00  & 1.02 (+2.0\%) & 1.04 (+4.0\%) & 1.13 (+13.0\%) \\
\texttt{cp}    & 75k & 11.86 & 12.45 (+5.0\%) & 12.36 (+4.2\%) & 14.44 (+21.8\%)  \\
\texttt{cut}   & 69k & 1.79  & 1.85 (+3.4\%) & 1.84 (+2.8\%) & 2.10 (+17.3\%) \\
\texttt{dd}    & 72k & 2.68  & 2.77 (+3.4\%) & 2.74 (+2.2\%) & 3.30 (+23.1\%) \\
\texttt{df}    & 70k & 13.86 & 14.55 (+5.0\%) & 14.65 (+5.7\%) & 18.09 (+30.5\%) \\
\texttt{du}    & 70k & 9.13  & 9.72 (+6.5\%) & 9.81 (+7.4\%) & 11.27 (+23.4\%) \\
\texttt{nohup} & 68k & 1.97  & 1.95 (-1.0\%) & 1.91 (-3.0\%) & 2.25 (+14.2\%) \\
\texttt{ptx}   & 71k & 4.20  & 4.49 (+6.9\%) & 4.45 (+6.0\%) & 5.11 (+21.7\%) \\
\texttt{tail}  & 72k & 3.44  & 3.62 (+5.2\%) & 3.68 (+7.0\%) & 4.18 (+21.5\%) \\
\midrule
Average &   &  & +4.0\% & +4.0\% & +20.7\% \\
\bottomrule
\end{tabular}
\label{tab:threadcomparison}
\end{table*}

We enhanced the \Goblint~\cite{DBLP:conf/tacas/SaanKPHESVS26} static analyzer to support
\pthreadOnce, barriers, and complex interactions between thread creation and mutex locking.
The implementations deviate slightly from the presentation here: they support \pthreads directly rather than idealized primitives
and use abstract digests in the sense of~\cite{hammerandnail}.
For thread-creation and mutex-locking interactions, they use mixed flow-sensitivity~\cite{DBLP:conf/fm/SeidlVES26}
to compute abstract digests at global unknowns.

We evaluate scalability and runtime on 60k--80k-LOC GNU Coreutils.
Their size precludes ground-truth annotations, but they allow measuring the runtime impact of the
new digests
against a minimal baseline and the \textsc{SV-Comp} configuration.
We ran each benchmark on Ubuntu 26.04 with a 2.5\,GHz Intel Core i7-11700,
a 5-minute timeout, and a 20\,GB memory limit.
Table~\ref{tab:threadcomparison} shows the cost of the enhanced synchronization features.
It is most pronounced
for the locked-creation feature. %
For \texttt{df}, locked
creation handling adds $30.5\%$ to the baseline time.
However, across $\approx$ 30,000 \textsc{SV-Comp} benchmarks enabling the enhanced
synchronization features increased total runtime by $\approx2\%$, which is negligible. %

We also compare our implementation with other static analyzers.
We gathered 90 small benchmarks, 20--70 lines of code each, testing the correct handling
of each of our four features: 28 \pthreadOnce, 26 barriers, 18 \creationLockset, and 18 \descendantLockset\ tests.
Each file contains one focused test annotated with the expected ground truth and sometimes motivated by a real-world use case.
Our main goal is to determine whether analyzers can correctly certify the absence of data races
in the presence of these features. A \textit{false positive} is an unexpected warning;
avoiding them is crucial for the acceptance of static analyzers in practice.
Conversely, a \textit{false negative} is a missed race;
\textit{even one compromises soundness}. We encode the desired behavior as
\textsc{SV-Comp} verdicts. %
We compared our implementation with %
\textsc{CPAChecker}~\cite{DBLP:conf/tacas/BaierBCJKLLRSWW24},
\textsc{Dartagnan}~\cite{DBLP:conf/tacas/LeonFH020},
\textsc{UGemcutter}~\cite{DBLP:conf/tacas/KlumppDHSEFP22},
\textsc{ESBMC-kind}~\cite{DBLP:journals/sttt/GadelhaIC17},
and \textsc{RacerF}~\cite{DBLP:conf/ecoop/DacikV25}. We also wanted to check whether
soundness concerns about~\cite{Lu2025} carry over to the implementation, but the code is unavailable
 and the authors did not respond to inquiries.

\cref{tab:tool-feature-comparison} summarizes the experiment, detailed in \appref{G}{app:experiments}.
The benchmark set reveals no false positives or false negatives in our implementation.
Since our implementation builds on abstract interpretation, it may produce false positives on other programs with the considered features.

All tools treat barriers in a sound way (which can be trivially achieved by not providing specific handling), they however do not exploit
the information completely to reduce false positives. \textsc{Dartagnan} and \textsc{UGemcutter} avoid false negatives but return
\texttt{ERROR} and \texttt{UNKNOWN}, respectively, when barriers or \texttt{pthread\_once} are present. For \texttt{pthread\_once},
\textsc{CPAChecker}, \textsc{RacerF}, and \textsc{ESBMC-kind} miss races, but
do not produce false positives except for \textsc{RacerF}. All tools miss races in the \creationLockset and \descendantLockset
categories, demonstrating unsound behavior. \textsc{Dartagnan}, \textsc{ESBMC-kind}, and \textsc{UGemcutter}
produce no false positives.
We hope contributing these tests to \SVCOMP will encourage tools to support these features
and spur competition.

\begin{table*}[t]
    \newcommand{\chec}{$\textcolor{teal}\checkmark$}
    \newcommand{\parti}{ $\textcolor{red}{\textbf{-}}$ }
    \newcommand{\nono} { $\textcolor{red}{\textbf{+}}$ }
    \newcommand{\resultpair}[2]{\makebox[1.2em][r]{#1}\,/\,\makebox[1.2em][l]{#2}}
	\centering
	\scriptsize
	\caption{Capability of static analyzers to establish data-race freedom in the presence of enhanced synchronization features.
	The table reports the number of false negatives (\parti, unsound) and false-positive race warnings (\nono) for each tool.}
	\begin{tabular}{@{}l@{\hspace{2em}}c@{\hspace{2em}}c@{\hspace{2em}}c@{\hspace{2em}}c@{}}
		\toprule
		Tool & \texttt{pthread\_once} & \texttt{barrier\_wait}
			 & \creationLockset & \descendantLockset \\
		     & \parti\,/\,\nono & \parti\,/\,\nono
		     & \parti\,/\,\nono & \parti\,/\,\nono \\
		\midrule
		\Goblint\ (this work)        & \resultpair{0}{0} & \resultpair{0}{0}  & \resultpair{0}{0} & \resultpair{0}{0} \\
		\textsc{CPAChecker}          & \resultpair{9}{0} & \resultpair{0}{12} & \resultpair{2}{1} & \resultpair{2}{0} \\
		\textsc{Dartagnan}           & \resultpair{0}{0} & \resultpair{0}{0}  & \resultpair{3}{0} & \resultpair{3}{0} \\
		\textsc{Ultimate Gemcutter}  & \resultpair{0}{0} & \resultpair{0}{0}  & \resultpair{2}{0} & \resultpair{2}{0} \\
		\textsc{ESBMC-kind}          & \resultpair{6}{0} & \resultpair{0}{4}  & \resultpair{5}{0} & \resultpair{3}{0} \\
		\textsc{RacerF}              & \resultpair{3}{7} & \resultpair{0}{7}  & \resultpair{1}{3} & \resultpair{2}{3} \\
		\bottomrule
	\end{tabular}
	\label{tab:tool-feature-comparison}
\end{table*}

\section{Related Work}\label{s:related}
Local traces and digests were proposed in a series of papers on thread-modular value analyses~\cite{schwarz2023clustered,schwarzthesis,hammerandnail,DBLP:conf/sas/SchwarzSSAEV21}.
Recent work~\cite{SchwarzE26} re-uses them for race detection.
We generalize both to cover more synchronization mechanisms, and
show that $||^?$ can check whether two actions
may run concurrently, which we use for barriers.

\citet{Terauchi08} presents an approach for framing race-freedom of programs as a typing problem, which in turn is
solved using rational linear programming.
It covers classical lock-based synchronization and mechanisms rarely supported by static analyses,
including signaling, semaphores, and read-write locks.
However, the approach differs from ours both in the used techniques, type-checking rather than abstract interpretation,
and in the considered synchronization mechanisms.

\paragraph{Barriers.} Barriers have been widely
studied~\cite{le_verification_2013,mittermayr_kronecker_2016,kokologiannakis_bam_2021,murthy_design_2016}, including
analyses of OpenMP-style barriers \cite{Lin05,SwainLLLG020,SwainH18}.
In OpenMP, all threads in a team share the same
code and must arrive at the syntactically same barrier, so code before a barrier cannot
run in parallel with code after it.
This does not carry over to \pthreads: different syntactic
calls to \textsf{wait} may belong together, and the capacity of a barrier
need not equal the number of running threads.
Here, we focus on work on \pthreads-style barriers.
\citet{mittermayr_kronecker_2016} extend the use of Kronecker algebra from the static analysis of semaphores to barriers.
\citet{kokologiannakis_bam_2021} speed up stateless model checking by merging interleavings that differ only in barrier-arrival order.
There, it is required that the initial capacity of a barrier is an upper bound to the number of \textsf{wait} calls that happen in parallel.
The analysis is justified w.r.t.\ a execution graph semantics.
\citet{murthy_design_2016} verify (using SPIN) a distributed phaser, a synchronization construct similar to barriers,
but do not analyze code using it.
\citet{le_verification_2013} propose a separation-logic using bounded partial permissions
to show that participating threads execute in lockstep phases separated by barriers, thereby excluding races across phases.
Their approach and others~\cite{HoborG2012} rely on user annotations; ours is fully automatic.
\citet{JeremiassenE94} split executions into barrier-delimited phases to determine which sections cannot run concurrently,
aiming to reduce compiler-introduced \emph{false sharing} and cache misses, while \citet{DBLP:journals/taco/DasSR15} do so to reduce the
instrumentation overhead of \emph{dynamic} approaches.

\paragraph{Feature \pthreadOnce.} For \pthreadOnce, missing support has been noted as a limitation~\cite{SiegelZLZMEDR15, LeRVBDO22};
to our knowledge, ours is the first analysis beyond syntactic checks~\cite{sonarsource},
both for race detection or static analysis more broadly.

\paragraph{Mutexes and Creation Orders.}
In \emph{static} analysis, research on the patterns from \cref{s:descendant-lockset} is sparse.
To our knowledge, only \citet{Lu2025} explicitly use it to exclude races,
via \emph{Segmented Thread-Sensitive Control Flow Graph}s.
This analysis is not justified w.r.t. a concrete semantics and is unsound on several litmus tests.
They also describe mutexes held throughout descendant-thread execution, but give no analysis to detect this pattern;
to our knowledge, ours is the first.
\emph{Dynamic} race-\emph{prediction} techniques~\cite{Smaragdakis12, Kini17, Genc19, Sulzmann20}
reason about ordering constraints induced by locks and other dependencies to determine
which re-orderings of an observed execution are possible.
This resembles our use of digests as synchronization histories to rule out parallel accesses,
but the analysis problem differs substantially.
Predictive analyses start from a single execution and seek to weaken its observed ordering without
admitting infeasible re-orderings; our sound static analysis overapproximates all executions and
must warn unless it can establish that the relevant accesses are ordered in \emph{every} execution,
not just those arising from re-orderings.
The descendant-lockset construction establishes the same must-happen-before relationship that
CP/WCP-style rules derive from lock synchronization --- if a thread is known to acquire a
specific mutex before another thread (in our case because the first thread
is statically known to be an ancestor of the second and to hold the mutex when the descendant is created),
then the release by the first thread must precede the acquire by the second thread.
Unlike predictive analyses, which reason over a single trace, our analysis establishes
this ordering soundly across all executions.
\citet{Sulzmann23} extend dynamic lockset reasoning with cross-thread critical sections,
in which an event may be protected by an acquire--release pair executed by another thread
when it must occur between them.
This captures the same pattern as our \emph{creationLocksets}:
a mutex held by an ancestor throughout a descendant's execution can protect accesses in the descendant,
although it never acquires the mutex.
Their analysis operates on observed executions and derives
cross-thread critical sections from must-happen-before information; our contribution is a sound
static abstraction establishing this protection across all executions using creation, join, thread-id,
and lock-history information.

\section{Conclusion and Future Work}\label{s:conclusion}
Sound static race detection comes with strong guarantees; reducing false positives is crucial for usability.
We built on digests, a framework refining thread-modular abstract interpretation,
by distinguishing program points according to thread histories.
While previous work cast existing techniques as digests to justify their soundness,
we present novel digests for several interesting concurrency primitives that have received
little to no attention in sound static race detection.
Future work may apply our techniques to other concurrency bugs, such as deadlocks and termination.

{
  \renewcommand{\doi}[1]{\textsc{doi}: \href{http://dx.doi.org/#1}{\nolinkurl{#1}}}
  \bibliographystyle{splncs04nat}
  \bibliography{bibliography}
}

\ifdefined\extended
\clearpage
\appendix
\section{Requirements on $||_{e_0, e_1}^?$}\label{app:mhp}

Formally, we demand that a predicate $||_{e_0, e_1}^?$ satisfies for all digests $A_0, A_1$:
\[
\begin{array}{ll}
  \exists t_0, t_1, t' \in \Traces:\\
  \qquad t_0 \textsf{ ends in } u_0 \land \alpha_\Digests\,t_0 = A_0 \land t_1 \textsf{ ends in } u_1 \land \alpha_\Digests\,t_1 = A_1 \land\\
  \qquad \oNone \neq \semT{(u_0,\act_0,u_0')}(t_0,\semT{(u_1,\act_1,u_1')}(t_1,t'))\land\\
  \qquad \oNone \neq \semT{(u_1,\act_1,u_1')}(t_1,\semT{(u_0,\act_0,u_0')}(t_0,t')) \implies \\[1.5ex]
  A_0\,||_{(u_0,\act_0,v_0), (u_1,\act_1,v_1)}^?\,A_1 = \top
\end{array}
\]
The definition in terms of bi-directional digest compatibility is then given by
\[
\begin{array}{ll}
  A_0\,||_{(u_0,\act_0,v_0), (u_1,\act_1,v_1)}^?\,A_1 =\\
  \qquad\begin{cases}
    \textsf{false} & \text{if } \neg\exists A':
          \oNone \neq \aSemDigest{u_0,\act_0}(A_0,\aSemDigest{u_1,\act_1}(A_1,A'))\land\\
        & \qquad\quad
       \oNone \neq \aSemDigest{u_0,\act_0}(A_1,\aSemDigest{u_1,\act_1}(A_0,A'))\\
    \top & \text{otherwise}
  \end{cases}
\end{array}
\]
The soundness of this definition follows directly from the admissibility of the digest.

\section{Proofs for~\cref{prop:data-race-bidirectional-compat,prop:refine-sound}}\label{app:prop1-prop2}
Recall \cref{prop:data-race-bidirectional-compat} from \cref{ss:lt}:
\begin{proposition}
    Two accesses to a global variable $g$ are racy in the intuitive sense if and only if
    at least one is a write and
    traces $t_0$ and $t_1$ ending in the respective CFG predecessors exist, s.t.\ they
        are bidirectionally trace compatible.
\end{proposition}
\begin{proof}
More formally, in line with \cite{SchwarzE26}, define two accesses to be \emph{racy in the intuitive sense}
if they access the same memory,
at least one is a write, and there exists a local trace in which they are unordered after disregarding the order induced
 by the accesses themselves. The proof is then a simplification of the proof of \cite[Theorem 1]{SchwarzE26},
 which uses atomicity mutexes to order accesses to globals.
 The existence of observing actions of arity n does not complicate the proof here,
 as it does not consider the \emph{reason} some actions are ordered w.r.t. each other.
 \qed
\end{proof}
\medskip
Recall \cref{prop:refine-sound} from \cref{ss:digests}:
\begin{proposition}\label{prop:refine-sound2}
    Refining a sound abstract interpretation of a concurrent program with an admissible digest yields
    a sound abstract interpretation.
\end{proposition}
\begin{proof}
    First, a version of the concrete semantics where the unknowns
    are split according to the digests is constructed.
    This refined concrete semantics is shown equivalent to the original concrete semantics,
    which is shown by giving a bijection between least solutions of either of the constraint systems.
    The property is then established by fixpoint induction performing a case
    distinction based on the considered constraint. The proof is a (slight)
    generalization of the proof of Theorem 4 of Section 2.3 of the PhD thesis
    of~\citet{schwarzthesis}, to account for the higher arity of observing actions but
    no additional complications arise here. Then, from the soundness of the
    base analysis, the soundness of the refined analysis follows by checking that
    by concretizing a solution of the abstract refined analysis one obtains a
    solution of the refined concrete constraint system. This is done by verifying
    this property for each constraint.\qed
\end{proof}
\medskip

Technically, the admissibility of the lockset digest would also need to be
re-established for the additional considered actions, but as they return
the digest unmodified, and \texttt{pthread\_once} and barriers do not affect sets
of held locks, this is straightforward.

\section{Details on  \& Extended Support for \pthreadOnce}\label{app:oncelong}
Recall proposition \cref{p:onceSound} from \cref{ss:once}.
\begin{proposition}
  The digest given in \cref{ss:once} is admissible, and its predicate $||^?$ is sound.
\end{proposition}
\begin{proof}
    To verify admissibility, i.e., conditions \labelcref{def:initSound,def:SemSound,def:NewSound,def:CreateNewConsistent}, we first need to
    define
    \[
        \alpha_\Digests\,t_0 = (A_0,C_0)
    \]
    where $A_0$ is the set of control variables $o$ for which
    $\onceStart(o)$ appears along the ego-lane of $t_0$, but $\onceEnd(o)$
    does not.
    $C_0$, on the other hand, is the set of control variables $o$ for
    which $\onceEnd(o)$ is backwards reachable by following
    program order edges $\to_{po}$,
    once order edges $\to_{o_i}$, and thread creation edges $\to_{tc}$
    in reverse.
    This corresponds to those $\onceEnd(o)$ actions that are either along the ego-lane of $t_0$,
    or along the ego-lane of one of its parents up to the point where the parent created
    the ego thread, or that the current thread has learned about by communicating
    through other control variables.
    This choice may seem complicated when contrasted with just considering all
    $\onceEnd(o)$ actions that appear in the local trace, but is deliberate: in this
    way, information about control variables is not propagated along other pairs
    of observing/observable actions, improving performance of the analysis.
    We can then check conditions \labelcref{def:initSound,def:SemSound,def:NewSound,def:CreateNewConsistent}
    by induction on the length of the trace.
    The definition of $\initDigests = \{ (\emptyset,\emptyset) \}$ in \cref{f:onceASimple} coincides
    with requirement \labelcref{def:initSound}, as in all initial traces, no $\onceStart(o)$ or $\onceEnd(o)$
    actions have been executed.

    We need to verify requirement \labelcref{def:SemSound} for all edges, which we recall is:
    \begin{equation*}
    \begin{array}{lll}
        \textsf{map}\;\alpha_\Digests\;(\semT{e}(t_0,\cdots,t_{k-1}))
        \hat{\sqsubseteq} \semDigests{u,\act}(\alpha_\Digests\,t_0,\cdots,\alpha_\Digests\,t_{k-1})
    \end{array}
    \label{def:SemSound2}
    \end{equation*}
    where $\alpha_\Digests$ is mapped over the trace option returned by $\semT{e}(t_0,\cdots,t_{k-1})$, and $\hat{\sqsubseteq}$
    is the partial order on $\Digests \option$ with $\oNone$ as bottom and all other elements incomparable.

    \medskip
    \noindent
    Let us first consider an
    edge $ e \equiv (u,\Pos(\onceRan(o)),v)$ in the trace. We have:
    \[
    \semDigests{\_, \Pos(\onceRan(o))}\,(A,C) =\\
        \qquad \begin{cases}
            \oNone & \text{if } \texttt{o} \not\in C\\
            \oSome (A,C) & \text{otherwise}
        \end{cases}
    \]
    Consider a trace $t_0$:
    \begin{itemize}
        \item If $\sem{e}\,t_0 = \oNone$, we have
        \[
            \textsf{map}\;\alpha_\Digests\;(\semT{e}\,t_0) =
            \textsf{map}\;\alpha_\Digests\;\oNone = \oNone\,
            \hat{\sqsubseteq} \semDigests{u,\Pos(\onceRan(o))}(\alpha_\Digests\,t_0)
        \]
        which holds as $\oNone$ is the bottom element of the partial order $\hat{\sqsubseteq}$.
        \item Consider now the case where $\sem{e}\,t_0 = \oSome t_1$ and let
            $\alpha_\Digests\,t_0 = (A_0,C_0)$.
            \begin{itemize}
                \item Consider the case where $o \not\in C_0$. Then, $\onceEnd(o)$ does not appear in
                $t_0$: By construction, $\onceStart(o)$ appears in $t_0$ as the guard can syntactically
                only appear between $\onceStart(o)$ and $\onceEnd(o)$.
                All actions $\onceStart(o)$ and $\onceEnd(o)$ as well as actions in between are totally
                ordered by $\to_o$ and $\to_{po}$.
                Thus, if $\onceEnd(o)$ appears anywhere in $t_0$ it is
                backwards reachable by following
                program order edges $\to_{po}$,
                once order edges $\to_{o_i}$,
                and thread creation edges $\to_{tc}$.
                As $\onceEnd(o)$ does not appear in $t_0$, the code associated with $o$
                has not been executed, and thus $\sem{e}\,t_0 = \oNone$, which is a contradiction.
                \item It remains to consider the case where $o \in C_0$.
                Appending an action $\Pos(\onceRan(o))$ to the trace does not change the set of
                active or completed control variables, and thus $\alpha_\Digests\,t_1 = (A_0,C_0)$.
                Thus, we have:
                \[
                    \begin{array}{lll}
                        &&\textsf{map}\;\alpha_\Digests\;(\semT{e}\,t_0) =
                        \textsf{map}\;\alpha_\Digests\;(\oSome t_1) =
                        \oSome (A_0,C_0)\\
                        &\hat{\sqsubseteq}&
                        \oSome (A_0,C_0) =
                        \semDigests{u,\Pos(\onceRan(o))}(\alpha_\Digests\,t_0)
                    \end{array}
                \]
            \end{itemize}
    \end{itemize}
    Thus, the claim follows.
    The proofs proceed analogously for $ e \equiv (u,\Neg(\onceRan(o)),v)$,
    as well as for edges not involving control variables.

    \medskip
    \noindent
    Next, consider an edge $e \equiv (u,\onceStart(o),v)$. Since
    $\onceStart(o)$ observes either $\onceInit(o)$ or $\onceEnd(o)$, its
    concrete semantics has two arguments. Let
    $\alpha_\Digests\,t_i=(A_i,C_i)$ for $i\in\{0,1\}$. We have
    \[
        \semDigests{\_,\onceStart(o)}\,(A_0,C_0)\,(A_1,C_1)
        =
        \begin{cases}
            \oSome(A_0\cup\{o\},C_0\cup C_1) & \text{if }o\notin A_0,\\
            \oNone & \text{otherwise}.
        \end{cases}
    \]
    Again, if $\semT{e}(t_0,t_1)=\oNone$, requirement~\labelcref{def:SemSound} follows
    immediately because $\oNone$ is the least element of
    $\hat{\sqsubseteq}$. It therefore remains to consider
    $\semT{e}(t_0,t_1)=\oSome t'$.
    In this case, $o\notin A_0$: otherwise, the ego-lane of $t_0$ would
    contain an unmatched $\onceStart(o)$. A further $\onceStart(o)$ on the
    same lane cannot be supplied with a predecessor in the once order.
    The observable action supplying the unmatched start has already been
    consumed, while every later $\onceEnd(o)$ is ordered after the new start.
    Adding the required once-order edge would therefore either violate its
    uniqueness requirement or create a cycle with program order. Hence, the
    concrete successor could not be a local trace, contradicting
    $\semT{e}(t_0,t_1)=\oSome t'$.
    Appending the new action makes $o$ active on the ego-lane and leaves the
    activity of every other control variable unchanged. Consequently, the
    active component of $\alpha_\Digests\,t'$ is $A_0\cup\{o\}$.
    Moreover, the new once-order edge makes every completed action that is
    backwards reachable in either input trace also backwards reachable in
    $t'$, and it introduces no other $\onceEnd$ action.
    Thus, its completed component is $C_0\cup C_1$. We obtain
    \[
        \begin{array}{lll}
            &&\textsf{map}\;\alpha_\Digests\;(\semT{e}(t_0,t_1))
              = \oSome(A_0\cup\{o\},C_0\cup C_1)\\
            &\hat{\sqsubseteq}&
              \oSome(A_0\cup\{o\},C_0\cup C_1)
              =\semDigests{\_,\onceStart(o)}\,(A_0,C_0)\,(A_1,C_1).
        \end{array}
    \]

    \medskip
    \noindent
    Consider now an edge $e \equiv (u,\onceEnd(o),v)$. This action is
    non-observing, and its digest semantics is
    \[
        \semDigests{\_,\onceEnd(o)}\,(A_0,C_0)
        = \oSome(A_0\setminus\{o\},C_0\cup\{o\}).
    \]
    If $\semT{e}(t_0)=\oNone$, requirement~\labelcref{def:SemSound} again follows from the
    minimality of $\oNone$. Otherwise, let
    $\semT{e}(t_0)=\oSome t'$ and
    $\alpha_\Digests\,t_0=(A_0,C_0)$. Appending $\onceEnd(o)$ closes the
    section for $o$ on the ego-lane, while leaving all other active sections
    unchanged. Hence, the active component of $\alpha_\Digests\,t'$ is
    $A_0\setminus\{o\}$. The appended action is itself backwards reachable
    from the maximal element of $t'$, so $o$ is completed; all previously
    completed control variables remain backwards reachable. Therefore, the
    completed component is $C_0\cup\{o\}$, and
    \[
        \begin{array}{lll}
            &&\textsf{map}\;\alpha_\Digests\;(\semT{e}(t_0))
              = \oSome(A_0\setminus\{o\},C_0\cup\{o\})\\
            &\hat{\sqsubseteq}&
              \oSome(A_0\setminus\{o\},C_0\cup\{o\})
              =\semDigests{\_,\onceEnd(o)}\,(A_0,C_0).
        \end{array}
    \]

    This concludes the proof for requirement \labelcref{def:SemSound}.

    \medskip
    \noindent
    We next verify requirement~\labelcref{def:NewSound}, i.e.,
    \[
        \textsf{map}\;\alpha_\Digests\;(\new\;t\;u\;u_1)
        \mathrel{\hat{\sqsubseteq}}
        \newDigests\;(\alpha_\Digests\,t)\;u\;u_1
    \]
    for every trace $t$ and all program points $u,u_1$. Let
    $\alpha_\Digests\,t=(A_0,C_0)$. By definition,
    \[
        \newDigests\;(A_0,C_0)\;u\;u_1
        = \oSome(\emptyset,C_0).
    \]
    If $\new\;t\;u\;u_1=\oNone$, the left-hand side of requirement~\labelcref{def:NewSound}
    is $\oNone$, so the claim follows from the minimality of $\oNone$.
    Otherwise, let $\new\;t\;u\;u_1=\oSome t'$. The ego-lane of the newly
    created thread in $t'$ consists only of its initial configuration.
    In particular, it contains no unmatched $\onceStart(o)$ action, and the
    active component of $\alpha_\Digests\,t'$ is therefore empty.
    The thread-creation edge makes exactly the part of the creating thread's
    history available at the creation point backwards reachable from the new
    thread. Consequently, a $\onceEnd(o)$ is backwards reachable in $t'$ if
    and only if it contributes $o$ to $C_0$. Thus,
    $\alpha_\Digests\,t'=(\emptyset,C_0)$, and we obtain
    \[
        \begin{array}{lll}
            &&\textsf{map}\;\alpha_\Digests\;(\new\;t\;u\;u_1)
              = \oSome(\emptyset,C_0)\\
            &\hat{\sqsubseteq}&
              \oSome(\emptyset,C_0)
              = \newDigests\;(\alpha_\Digests\,t)\;u\;u_1.
        \end{array}
    \]
    This proves requirement~\labelcref{def:SemSound}.

    \medskip
    \noindent
    Finally, requirement~\labelcref{def:CreateNewConsistent} demands that, for every digest
    $(A_0,C_0)$,
    \[
        \semDigests{u,x=\create(u_1)}\,(A_0,C_0)\neq\oNone
        \quad\Longrightarrow\quad
        \newDigests\;(A_0,C_0)\;u\;u_1\neq\oNone.
    \]
    This follows directly from the definition of $\newDigests$: for every
    $(A_0,C_0)$ and all program points $u,u_1$,
    \[
        \newDigests\;(A_0,C_0)\;u\;u_1
        =\oSome(\emptyset,C_0)\neq\oNone.
    \]
    Thus, the conclusion of the implication holds independently of its
    premise, which proves requirement~\labelcref{def:CreateNewConsistent} and concludes the proof of
    admissibility.
    \medskip
    \noindent

  For the proof of soundness of $||^?$,
  we observe that actions $\onceStart(o)$ and $\onceEnd(o)$ and the actions appearing in between in any local
  trace are totally ordered by $\to_{o}$ and the program order $\to_{po}$.
  W.l.o.g., assume $(A \cap (A' \cup C') \neq \emptyset)$ and that the
  considered access by thread with digest $(A, C)$ is along an edge $e$.
  Consider $t \in \gamma_\Digests (A, C)$ and $t' \in \gamma_\Digests (A',C')$.
  If $o$ is in $A$ and in $A'$, for both traces to belong to the same execution, either $t$ is a subtrace
  of $t'$, or $t'$ is a subtrace of $t$, as $\onceEnd(o)$ needs to occur before a second thread can execute
  $\onceStart(o)$. Thus, $t$ and $t'$ cannot be bidirectionally trace-compatible.
  Now consider the case where $o \in A \cap C'$. Then, in $t'$, there is an edge $\Neg(\onceRan(o))$, meaning
  the code associated with $o$ was actually run inside of $t'$. For an access to occur with $o \in A$,
  $\Neg(\onceRan(o))$ needs to be contained in the ego-lane of $t$. As this holds for at most one thread
  in any execution, $t$ and $t'$ can only be compatible if they coincide up to the
  execution of the access. Thus, $t$ is a subtrace of $t'$, and
  $\sem{e}(t_0,t_1) = \oNone$ and the traces are not bidirectionally compatible.\qed
\end{proof}%
\medskip
Beyond supporting \pthreadOnce for checking for data races, one can also increase the precision of an analysis
by only reading those values that globals may have after the code inside \pthreadOnce has been executed.

A digest to help analyze this feature (\cref{f:onceA}) may then track \emph{three} pieces of information in a tuple $(A,C,O)$:
\begin{itemize}
    \item Some set $A$ of \emph{active} control variables as in \cref{ss:once}.
    \item Some set $C$ of \emph{completed} control variables, i.e., control variables for which the associated fragment
    has surely been run. Unlike in \cref{ss:once}, this is not only the case if $\onceEnd(o)$ was called by the ego
    thread or one of its parents, but also when the ego thread incorporates an observable action where the control variable
    is completed.
    \item Additionally, a map $O$ tracking for each mutex $a$, for which of the control variables $o$ the mutex $a$ was unlocked
    while $o$ was active.
\end{itemize}
The interesting case is the right-hand side for acquiring a mutex $a$: In case the ego thread has observed that the mutex $a$
was unlocked while some control variable $o$ was active, but the trace ending in the unlock of $a$ has neither observed
this nor is the control variable currently active, there must have been a later unlock of $a$, where
the control variable was active. Thus, no two traces with such digests can be compatible.
\begin{figure*}
    \begin{minipage}[t]{.4\linewidth}\[
        \begin{array}{lll}
        \Digests = 2^\Onces\times 2^\Onces \times (\Mutexes \to 2^{\Onces} )\\[0.5ex]
		\initDigests = \{(\emptyset,\emptyset,\{ a \mapsto \emptyset \mid a \in \Mutexes \})\}\\[0.5ex]
        \newDigests\,(A,C,O)\,u\,u_1= \oSome (\emptyset,C,O)\\[0.5ex]
        \semDigests{\_, \Pos(\onceRan(o))}\,(A,C,O) =\\
        \qquad \begin{cases}
            \oNone & \text{if } \texttt{o} \not\in C\\
            \oSome (A,C,O) & \text{otherwise}
        \end{cases}\\
        \semDigests{\_, \Neg(\onceRan(o))}\,(A,C,O) =\\
        \qquad \begin{cases}
            \oNone & \text{if } \texttt{o} \in C\\
            \oSome (A,C,O) & \text{otherwise}
        \end{cases}
        \end{array}%
    \]
    \end{minipage}%
    \begin{minipage}[t]{.6\linewidth}\[
        \begin{array}{lll}
            \semDigests{\_, \lock(a)}\,(A_0,C_0,O_0)\,(A_1,C_1,O_1) =\\
            \qquad
            \begin{cases}
                \oNone & \text{if } O_0\,a \not\subseteq (C_1 \cup A_1)\\
                \oSome (A_0, C_0 \cup C_1, O_0 \cup O_1) & \text{otherwise}
            \end{cases}\\[4ex]
            \semDigests{\_, \unlock(a)}\,(A,C,O) = \oSome (A, C, O \oplus \{ a \mapsto O\,a \cup A \})\\[2ex]

            \semDigests{\_, \act}\,(A_0,C_0,O_0)\,(A_1,C_1,O_1) =\\
            \qquad \oSome (A_0, C_0 \cup C_1, O_0 \cup O_1) \quad \text{(other observing)}\\[2ex]
            \semDigests{\_, \act}\,(A,C,O) = \oSome (A,C,O) \quad \text{(other non-observing)}\qquad\;\;
        \end{array}
        \]
    \end{minipage}\\
    \begin{minipage}[t]{\linewidth}\[
        \begin{array}{lll}
        \semDigests{\_, \onceStart(o)}\,(A_0,C_0,O_0)\,(A_1,C_1,O_1) =
        \begin{cases}
            \oSome (A_0 \cup \{o\}, C_0 \cup C_1, O_0 \cup O_1) & o \not\in A_0\\
            \oNone & \text{otherwise}
        \end{cases}\\
        \semDigests{\_, \onceEnd(o)}\,(A,C,O) = \oSome (A \setminus \{o\}, C \cup \{o\}, O)\\

        \end{array}%
    \]
    \end{minipage}\\
    \caption{Right-hand sides for expressing \pthreadOnce as a refinement, where $\cup$ on maps
    is defined as the point-wise union of the right-hand sides.}\label{f:onceA}
\end{figure*}
\begin{example}
    Consider again the program in \cref{f:once}, and an analysis associating values of globals with mutexes that are
    always held when the global is accessed, refined with the digests from \cref{f:onceA}.
    At the later $\lock(m_{dev})$ in both of the threads, the only reaching digests of the ego threads are
    $(\emptyset,\{o\},\{ m_{dev} \mapsto \{ o \}\})$. Thus, only those unknowns associated with observable actions with a digest
    $(A_1,C_1,O_1)$ where $o$ is either in $A_1$ or $C_1$ can be incorporated here. For all these unknowns, the value of $g$
    is known to not be null. Thus, the assertions can be shown in both threads.
    $\hfill\qed$
\end{example}

\section{Details on Barriers}\label{sec:barriers-app}
We first remark on the difference between barriers and signals.
\begin{remark}
    While the signalling mechanism of \pthreads\ at first glance seems similar to barriers of capacity one,
    the information that
    a call to \texttt{signal\_wait(s)} has returned cannot help exclude data races.
    \pthreads allows for spurious wakeups, so such a call may return even though $s$
    has not been signalled.
\end{remark}
We next turn to the soundness of using this digest for excluding races.
Recall that the barrier digest is built on a digest $\Digests_1$. We assume
that $\Digests_1$ is admissible and that its predicate
$||^?_{\Digests_1}$ is sound. Recall that in \cref{ss:digests} we have established
that $||^?_{\Digests_1}$ can be used to check whether arbitrary actions
may-happen-in-parallel.
We further assume that the base abstraction is insensitive to
the newly introduced barrier record and check actions. That is, extending a
trace by either action does not change its $\Digests_1$ abstraction, and the
corresponding base right-hand sides are identity functions. As required for the
generic thread IDs introduced in Section~2, $\pTID$ is constant along each
concrete thread and $\unique(q)$ implies that at most one concrete thread in
an execution has projected ID $q$.

Recall the following proposition from \cref{ss:barriers}:
\begin{proposition}
  Provided the digest $\NDigests{1}$ is admissible, and its predicate $||^{?}_{\NDigests{1}}$ is sound,
  the digest given in \cref{ss:barriers} is admissible, and its predicate $||^?$ sound.
\end{proposition}
\begin{proof}
  Even though the definition of the predicate $||^?$ appears in the right-hand side of the digests,
  the proof can proceed step-by-step, first showing admissibility and turning to the soundness of $||^?$
  later.
    Let $\alpha_{\Digests_1}$ be the abstraction function of the base digest.
    We define the abstraction function for
    the product digest by
    \[
        \alpha_\Digests\,t
        =(\alpha_{\Digests_1}\,t,R(t),O(t)).
    \]
    Here, $R(t)$ contains exactly those barriers $b$ for which a
    $\textsf{barrier\_N\_record}(b)$ action occurs on the ego-lane of $t$.
    For a projected thread ID $q$, the set $O(t)\,q$ contains $b$ if $t$
    contains a $\textsf{barrier\_N\_record}(b)$ action performed by a thread
    whose base digest at that action projects to $q$. Thus, $R$ describes
    records of the ego thread, whereas $O$ also includes records learned
    from incorporated traces. We verify requirements~\labelcref{def:initSound,def:SemSound,def:NewSound,def:CreateNewConsistent}.

    \medskip
    \noindent
    For requirement~\labelcref{def:initSound}, an initial trace contains no barrier action.
    Therefore, $R(t)=\emptyset$ and $O(t)$ is the everywhere-empty map for
    every $t\in\init$. Its base component belongs to
    $\initNDigests{1}$ by admissibility of $\Digests_1$. Consequently,
    \[
        \{\,\alpha_\Digests\,t\mid t\in\init\,\}
        =
        \{\,(A,\emptyset,\emptyset)
            \mid A\in\initNDigests{1}\,\}
        =\initDigests.
    \]

    \medskip
    \noindent
    We next verify requirement~\labelcref{def:SemSound}. As before, if a concrete right-hand side
    returns $\oNone$, the requirement follows immediately because $\oNone$
    is the least element of $\hat{\sqsubseteq}$. In the following cases, we
    may therefore assume that the concrete operation succeeds.

    First, let
    \[
        e=(u,\textsf{barrier\_N\_record}(b),v),\qquad
        \semT{e}(t_0)=\oSome t',
    \]
    and write $\alpha_\Digests\,t_0=(A_0,R_0,O_0)$. The record action is
    appended to the ego-lane. Barrier actions do not affect the base digest,
    so the first component remains $A_0$. The new record adds $b$ to $R_0$.
    It is performed by the ego thread, whose projected ID is
    $\pTID\,A_0$, and therefore also adds $b$ to the corresponding binding
    of $O_0$. No other binding changes. Hence,
    \[
    \begin{split}
        \alpha_\Digests\,t'
        =\bigl(&A_0,R_0\cup\{b\},\\
               &O_0\oplus
               \{(\pTID\,A_0)\mapsto
                   (O_0(\pTID\,A_0))\cup\{b\}\}\bigr),
    \end{split}
    \]
    which is exactly
    $\aSemDigest{\_,\textsf{barrier\_N\_record}(b)}
      (A_0,R_0,O_0)$.

    Now consider
    \[
        e=(u,\textsf{barrier\_N\_check}(b),v),\qquad
        \semT{e}(t_0,\ldots,t_{N-1})=\oSome t',
    \]
    and let
    $\alpha_\Digests\,t_i=(A_i,R_i,O_i)$.
    Concrete success means that the $N$ traces ending in the corresponding
    record action are distinct, mutually compatible, and can occur in
    parallel. Soundness of $||^?_{\Digests_1}$ therefore gives
    \[
        A_i||^?_{\Digests_1}A_j=\top
        \qquad\text{for every }i\neq j.
    \]
    It also implies the second check in the abstract right-hand side. To see
    this, fix $i\neq j$ and suppose that $\pTID\,A_i$ is unique. Every
    barrier record in $O_j(\pTID\,A_i)$ was then performed by the unique
    ego thread represented by $A_i$. If such a record were not in $R_i$,
    $t_j$ would already contain a strict continuation of the ego-lane of
    $t_i$. The two record actions supplied to the same barrier check could
    then not occur in parallel: combining the traces would either duplicate
    a barrier-order predecessor or create a causality cycle. This would
    contradict the assumed success of the concrete check. Therefore,
    \[
        \unique(\pTID\,A_i)
        \quad\Longrightarrow\quad
        O_j(\pTID\,A_i)\subseteq R_i.
    \]
    Thus, the guard of the successful branch of the abstract right-hand
    side holds.

    The check action adds no barrier record to its ego-lane, so the resulting
    $R$ component is $R_0$. It incorporates all input traces, and hence a
    record is available in the result exactly when it is available in at
    least one input. Its observation map is consequently the pointwise union
    $\bigcup_{i=0}^{N-1}O_i$. Finally, the check does not alter the base
    digest of the ego thread. We conclude that
    \[
        \alpha_\Digests\,t'
        =(A_0,R_0,\bigcup_{i=0}^{N-1}O_i),
    \]
    which agrees with the successful branch of
    $\aSemDigest{\_,\textsf{barrier\_N\_check}(b)}$.

    It remains to consider an edge $e=(u,\act,v)$ for any other $n$-ary
    action. Let
    $\alpha_\Digests\,t_i=(A_i,R_i,O_i)$ and assume
    $\semT{e}(t_0,\ldots,t_{n-1})=\oSome t'$. By requirement~\labelcref{def:SemSound} for the
    base digest and the flat order on options, its abstract right-hand side
    returns precisely the base component $A'$ of $t'$:
    \[
        \aSemNDigest{1}{u,\act}(A_0,\ldots,A_{n-1})
        =\oSome A'.
    \]
    The action adds no barrier record to the ego-lane, so the second
    component remains $R_0$. All records incorporated from the argument
    traces are captured by the pointwise union of their observation maps.
    Therefore,
    \[
        \alpha_\Digests\,t'
        =(A',R_0,\bigcup_{i=0}^{n-1}O_i),
    \]
    exactly as returned by the product right-hand side. This proves
    requirement~\labelcref{def:SemSound} in all cases.

    \medskip
    \noindent
    For requirement~\labelcref{def:NewSound}, let
    $\alpha_\Digests\,t=(A_0,R_0,O_0)$. If
    $\new\;t\;u\;u_1=\oNone$, the claim follows from the minimality of
    $\oNone$. Otherwise, let
    $\new\;t\;u\;u_1=\oSome t'$. Requirement~\labelcref{def:NewSound} for the base digest yields
    \[
        \newNDigests{1}\;A_0\;u\;u_1=\oSome A'
        \quad\text{with}\quad
        \alpha_{\Digests_1}\,t'=A'.
    \]
    The new thread's ego-lane initially contains no barrier record, so its
    $R$ component is empty. Through the creation edge, it inherits all
    records known at the creation point, and hence its observation map is
    $O_0$. It follows that
    \[
        \alpha_\Digests\,t'=(A',\emptyset,O_0)
        =\newDigests\;(A_0,R_0,O_0)\;u\;u_1,
    \]
    establishing requirement~\labelcref{def:NewSound}.

    \medskip
    \noindent
    Finally, suppose the premise of requirement~\labelcref{def:CreateNewConsistent} holds:
    \[
        \aSemDigest{u,x=\create(u_1)}(A_0,R_0,O_0)\neq\oNone.
    \]
    By the definition of the product right-hand side, the corresponding base
    right-hand side is also different from $\oNone$. Requirement~\labelcref{def:CreateNewConsistent} for
    $\Digests_1$ then gives
    \[
        \newNDigests{1}\;A_0\;u\;u_1=\oSome A'
    \]
    for some $A'$. By definition of the product operation,
    \[
        \newDigests\;(A_0,R_0,O_0)\;u\;u_1
        =\oSome(A',\emptyset,O_0)\neq\oNone.
    \]
    This proves requirement~\labelcref{def:CreateNewConsistent}, and hence the barrier digest is admissible.
    \medskip
    \noindent
  Let us now turn to the soundness of $||^?$:
  The first check is sound by soundness of $||^{?}_{\NDigests{1}}$. We show that the second
  argument for $\sqcap$ returning $\textsf{false}$ implies that the traces are not bi-directionally
  compatible.
  W.l.o.g., assume $(\unique\,(\pTID\,A_0) \land O_1\,(\pTID\,A_0) \not\subseteq R_0)$ and that the
  considered access by thread with digest $(A_0,R_0,O_0)$ is along an edge $e$.
  Consider $t_0 \in \gamma_\Digests (A_0,R_0,O_0)$ and $t_1 \in \gamma_\Digests (A_1,R_1,O_1)$.
  Then, trace $t_1$
  has observed an action performed by the ego-thread of $t_0$ that itself is not part of $t_0$ yet. Thus,
  a super-trace of $t_0$ containing the access and the observed action is a subtrace of $t_1$. Thus,
  $\sem{e}(t_0,t_1) = \oNone$ and the claim follows.\qed
\end{proof}

\section{Digests for Descendant Lockset}\label{app:descendant-lockset-long}

A digest helping with analyzing the pattern described in Section \ref{s:descendant-lockset} tracks the following information in a triple $((A,S), H, D)$:
\begin{itemize}
  \item A tuple $(A,S)$ holding the value of the thread \emph{id} digest $A$ and the value of the lockset digest $S$
  \item A map $H$, where each mutex $m$ is mapped to the set of thread \emph{id}s. This set includes only thread \emph{id}s, in which $m$ was locked. In \cref{e:descendant-ls}, $H$ holds the relevant information we need at the access to \texttt{global} in $i_1$.
  \item A map $D$, where we map each thread \emph{id} $i$ to a set of mutexes. If $m$ is included in the set, the conditions 2 and 3 listed in section \ref{s:descendant-lockset}.
\end{itemize}

\begin{figure*}
\[
\begin{array}{l}
\Digests = \NDigests{T\times S} \times (\Mutexes\to2^\TIDs) \times (\TIDs \to 2^\Mutexes)\\[1ex]

\newDigests\,((A,S),H,D)\,u\,u_1=\\
\quad\begin{cases}
  \oSome\,((A',S'),H,\{i\mapsto \Mutexes\mid i\in \TIDs\}) & \mathrm{if}\, \newNDigests{T\times S}\,(A, S)\,u\,u_1=\oSome\,(A',S')\\
  \oNone & \mathrm{otherwise}
\end{cases}\\[1ex]

\initDigests=
\{((A,S),\emptyset,\{i\mapsto \Mutexes\mid i\in\TIDs\})
  \mid (A,S)\in\initNDigests{T\times S}\}\\[1ex]

\aSemDigest{u, \create\,u_0}\,((A,S),H,D)=\\
\quad\begin{cases}
  \begin{aligned}
    &\Let\,\oSome A_\mathrm{new}=\newNDigests{T}\,A\,u\,u_0\,\In\\
    &\Let\,i_{\mathrm{ego}}=\pTID\,A\,\In\\
    &\Let\,I_\mathrm{new}^*=\{i\in\TIDs\mid
      \pTID\,A_\mathrm{new}\in (\{i\}\cup\mustanc\,i)\}\,\In\\
    &\Let\,D'=D\oplus \{i^*_\mathrm{new}\mapsto S\cap D\,i^*_\mathrm{new}
      \mid i^*_\mathrm{new}\in I^*_\mathrm{new},
      i_\mathrm{ego}\in\mustanc\,i^*_\mathrm{new}\}\,\In\\
    &\oSome\,((A',S),H,D')
  \end{aligned}
  & \mathrm{if}\ \aSemNDigest{T}{u,\create\,u_0}\,A = \oSome A'\\
  \oNone & \mathrm{otherwise}
\end{cases} \\[1ex]

\aSemDigest{u,\lock(a)}\,((A_0,S_0),H_0,D_0)\,((A_1,S_1),H_1,D_1)=\\
\quad\begin{cases}
  \oSome\,((A_0,S'),H_0\oplus\{a\mapsto (H_0\,a)\cup\{\pTID\,A_0\}\},D_0)
    & \mathrm{if}\ \aSemNDigest{S}{u,\lock(a)}\,S_0\,S_1=\oSome S'\\
  \oNone & \mathrm{otherwise}
\end{cases}\\[1ex]

\aSemDigest{u,\unlock(a)}\,((A,S),H,D)=\\
\quad\Let\,D'=\{i\mapsto D\,i\setminus\{a\}\mid i\in\TIDs\}\,\In\\
\quad\begin{cases}
  \oSome\,((A,S'),H,D')&\mathrm{if} \aSemNDigest{S}{u,\unlock(a)}\,S =\oSome S'\\
  \oNone & \mathrm{otherwise}
\end{cases}\\[1ex]

\aSemDigest{\_, \act}(((A_0,S_0),H_0,D_0),\dots,((A_{N-1},S_{N-1}),H_{N-1}, D_{N-1})) =\\
    \quad \begin{cases}
    \oSome ((A',S'),H_0,D_0) & \text{if } \aSemNDigest{T\times S}{\_,\act}((A_0,S_0),\dots,(A_{N-1},S_{N-1})) =\\&\quad \oSome (A',S')\\
    \oNone & \text{otherwise} \qquad\qquad\qquad\qquad\qquad\qquad \text{(other)}
    \end{cases}
\end{array}
\]
\caption{Descendant lockset digest, where $\oplus$ denotes updating a map. Missing bindings map to the empty set.}
\label{f:descendant-lockset-digest}
\end{figure*}

We then define the race check as:
\[
\begin{array}{l}
  (A_0,S_0,H_0, D_0)\, ||^{?}\, (A_1,S_1,H_1,D_1) =\Let\, i_0 = \pTID\,A_0\,\In\,\Let\, i_1 = \pTID\,A_1\,\In\\
  \begin{cases}
    \false &\mathrm{if\ }
    \begin{array}{l}
      ((\{i_1\}\cup\mustanc\,i_1)\cap \maycreate\,A_0\ne\emptyset\;\land\\
      \hphantom{(}(\exists i_{H_1}\in \TIDs, s\in\Mutexes: i_{H_1}\in H_1\,s\land s\in D_0\,i_1\land i_0\in\mustanc\,i_{H_1}))\;\lor\\
      \color{darkgray}((\{i_0\}\cup\mustanc\,i_0)\cap \maycreate\,A_1\ne\emptyset\;\land\\
      \color{darkgray}\hphantom{(}(\exists i_{H_0}\in \TIDs, s\in\Mutexes: i_{H_0}\in H_0\,s\land s\in D_1\,i_0\land i_1\in\mustanc\,i_{H_0}))
    \end{array}
    \\
    \top &\mathrm{otherwise}
  \end{cases}
\end{array}
\]

\begin{example}
  Considering the program from \cref{f:descendant-ls-example} again, the digests have the following interesting values at the accesses to \texttt{global}:
  \begin{itemize}
    \item At the access in $i_1$, the $H$ component has value \{$m\mapsto \{\maintid,i_1\}\}$.
    \item At the access in $\maintid$, the $D$ component has value $\{i_1\mapsto\{m\}, i_2\mapsto\{m\}\}$.
  \end{itemize}
  With these values and the may-race predicate given above, we can exclude the possibility of a data race.
\end{example}

\begin{proposition}
  Provided the digest $\NDigests{T\times S}$ is admissible, the digest given above is admissible and its definition for $||^?$ is sound.
\end{proposition}
\begin{proof}
  For proof of admissibility, we inductively verify that the digest fulfills
  \labelcref{def:initSound,def:SemSound,def:NewSound,def:CreateNewConsistent}.
  Let $\alpha_B$ be the abstraction function for $\NDigests{T\times S}$  and define
  \[
      \alpha_\Digests\,t=(\alpha_B\,t,H(t),D(t)).
  \]
  The auxiliary components are defined inductively along the creation-extended ego lane of the local trace.
  For traces in $\init$, $H$ maps every
  mutex to $\emptyset$, while $D$ maps every thread ID to $\Mutexes$.
  Executing $\lock(a)$ in a thread with projected ID $i$ adds $i$ to
  $H(a)$. A create action intersects $D(i)$ with the current lockset for
  precisely the possible descendant IDs selected by
  $I_{\mathrm{new}}^*$ in \cref{f:descendant-lockset-digest}. Executing
  $\unlock(a)$ removes $a$ from every entry of $D$. Other actions leave
  both components unchanged. A new thread inherits $H$ from its creator
  and starts with a map mapping all thread \emph{id}s to $\Mutexes$ for $D$.

  For requirement~\labelcref{def:initSound}, an initial trace contains no lock, unlock, or
  creation action. Hence, for every $t\in\init$,
  \[
      \alpha_\Digests\,t
      =((A,S),\emptyset,\{i\mapsto\Mutexes\mid i\in\TIDs\}) \in \initDigests
  \]
  where $(A,S)\in\initNDigests{T\times S}$. The converse inclusion follows
  from requirement~\labelcref{def:initSound} for the base digest, proving equality with the
  displayed $\initDigests$.

  We verify requirement~\labelcref{def:SemSound} by cases. A concrete result $\oNone$ is
  below every abstract result, so only successful concrete operations
  require consideration.

  Let $e\equiv(u,\create\,u_0,v)$,
  $\semT{e}(t)=\oSome t'$, and
  $\alpha_\Digests\,t=((A,S),H,D)$. Base admissibility and flatness of the
  option order give
  \[
      \aSemNDigest{T}{u,\create\,u_0}A=\oSome A',
      \qquad \alpha_B\,t'=(A',S).
  \]
  By requirement~\labelcref{def:CreateNewConsistent} for the
  thread-ID digest we have
  $\newNDigests{T}A\,u\,u_0=\oSome A_{\mathrm{new}}$ for some
  $A_{\mathrm{new}}$. Thus, the pattern match defining
  $A_{\mathrm{new}}$ in the successful branch is total. By the definition
  of $D$, the new creation adds the current-lockset condition for exactly
  the IDs in the resulting $I_{\mathrm{new}}^*$, so
  \[
      D'=D\oplus
      \{i\mapsto S\cap D(i)
        \mid i\in I_{\mathrm{new}}^*,
             \pTID\, A\in\mustanc\, i\}.
  \]
  A create performs no lock, hence $H$ is unchanged and
  $\alpha_\Digests\,t'=((A',S),H,D')$, coinciding with the result of the right-hand side.

  Next, consider
  $\semT{(u,\lock(a),v)}(t_0,t_1)=\oSome t'$ and let
  $\alpha_\Digests\,t_k=((A_k,S_k),H_k,D_k)$. The right-hand side of the base digest yield
  $\aSemNDigest{S}{u,\lock(a)}S_0=\oSome S'$ and leaves $A_0$
  unchanged. The action records that the ego thread, whose ID is
  $\pTID\, A_0$, has locked $a$; it does not modify $D$. Therefore,
  \[
      \alpha_\Digests\,t'
      =((A_0,S'),
        H_0\oplus\{a\mapsto H_0(a)\cup\{\pTID\, A_0\}\},D_0),
  \]
  agreeing with the result of the right-hand side.

  For a successful unlock operation,
  $\semT{(u,\unlock(a),v)}(t)=\oSome t'$ with
  $\alpha_\Digests\,t=((A,S),H,D)$, base admissibility yields
  $\aSemNDigest{S}{u,\unlock(a)}S=\oSome S'$.
  This unlock invalidates the uninterrupted-holding condition for $a$
  for every potential descendant ID, while changing no $H$ entry. Thus,
  \[
      \alpha_\Digests\,t'
      =((A,S'),H,\{i\mapsto D(i)\setminus\{a\}\mid i\in\TIDs\}),
  \]
  agreeing with the result of the right-hand side.

  For every other $n$-ary action, requirement~\labelcref{def:SemSound} for the base digest and
  flatness give the exact successful base successor $(A',S')$. Such an
  action changes neither auxiliary component by definition, so the
  product successor is $((A',S'),H_0,D_0)$, as specified by the
  ``other'' case in \cref{f:descendant-lockset-digest}. This proves requirement~\labelcref{def:SemSound}.

  For requirement~\labelcref{def:NewSound}, let
  $\alpha_\Digests\,t=((A,S),H,D)$.
  It once again suffices to consider only the case where the edge succeeds in the concrete.
  For
  $\new\;t\;u\;u_1=\oSome t'$, base requirement~\labelcref{def:NewSound} and flatness yield
  \[
      \newNDigests{T\times S}(A,S)\,u\,u_1=\oSome(A',S'),
      \qquad \alpha_B\,t'=(A',S').
  \]
  By definition, the child inherits $H$, has not itself established any
  creation constraint, and therefore starts with the everywhere-$\Mutexes$ map. Hence,
  \[
      \alpha_\Digests\,t'
      =((A',S'),H,\{i\mapsto\Mutexes\mid i\in\TIDs\}),
  \]
  which coincides with the definition of $\newDigests$.

  Finally, if the premise of requirement~\labelcref{def:CreateNewConsistent} holds for
  $((A,S),H,D)$, by
  requirements~\labelcref{def:SemSound,def:CreateNewConsistent} makes
  $\newNDigests{T\times S}(A,S)\,u\,u_1$ different from $\oNone$.
  The definition of the product new-thread operation then returns
  $\oSome((A',S'),H,\{i\mapsto\Mutexes\mid i\in\TIDs\})$ for the
  corresponding base result $(A',S')$, proving requirement~\labelcref{def:CreateNewConsistent}.
  Thus, the descendant-lockset digest is admissible.

  We now turn to the proof of soundness of $||^?$:
  W.l.o.g., assume the first disjunct of the $\false$ case holds.
  Consider two local traces $t$ and $t'$ where $t\in\gamma_\Digests(A_0,S_0,H_0, D_0)$
  and $t'\in\gamma_\Digests(A_1,S_1,H_1, D_1)$.

  Since $(\{i_1\}\cup\mustanc\,i_1)\cap \maycreate\,A_0\ne\emptyset$, $t$ has observed at least one
  transitive creation of threads associated with $i_1$---the thread \emph{id} at the end of $t'$.
  As $s\in D_0\, i_1$, all of those creations happened while $s$ was in the lockset and none of the
  creations is followed by $\unlock(s)$ in $i_0$.

  Since $i_{H_1}\in H_1\,s$, $t'$ has observed $\lock(s)$ in $i_{H_1}$. No matter which of the transitive creations
  in $i_0$ described above resulted in the creation of the thread at the end of $t'$---this lock happens after that creation,
   as $i_0\in \mustanc\,i_{H_1}$. The lock must also happen after a succeeding $\unlock(s)$ following the creation in $i_0$,
   since $s$ is in the lockset at that creation. Since no such unlock is included in $t$ whereas the later lock is
   included in $t'$, compatibility can hold only if $t$ is a subtrace of $t'$. Hence, the traces are not bidirectionally compatible.\qed
\end{proof}

\section{Digests for Creation Locksets}\label{app:creation-lockset-long}

We propose the digest in \cref{f:descendant-ls-digest} to analyze the conditions given at the end of \cref{s:creation-lockset}.
Such a digest is a triple $((A,S,J),C,O)$:
\begin{itemize}
  \item A triple $(A,S,J)$ representing the thread \emph{id} and lockset digest introduced in \cref{ss:digests}, and the join digest by~\citet{hammerandnail}.
  \item A map $C$ mapping each thread \emph{id} $i$ to a set of mutexes, where $m$ is included if conditions 2 and 3 are satisfied.
  \item A map $O$ mapping each thread \emph{id} $i$ to a pair. In the
  digest of the ego thread, $O\,i$ stores the $C$ component and the set of
  may-created thread \emph{id}s of the most recently observed
  sub-trace belonging to a thread represented by $i$. Such an observation occurs
  when the ego thread locks a mutex after that thread unlocks it. We use $O$
  to compare the lockset at a program point with inter-threaded mutex
  protections of other threads.
\end{itemize}

\begin{figure*}
  \[
    \begin{array}{l}
      \Digests = \NDigests{T\times S\times J} \times (\TIDs \to 2^\Mutexes) \times (\TIDs \to ((\TIDs\to 2^\Mutexes)\times 2^\TIDs))\\[1ex]

      \newDigests\,((A,S,J),C,O)\,u\,u_1=\\
      \quad \Let\, O'=\{i\mapsto(\emptyset,\emptyset)\mid i\in \TIDs\}\,\In\\
      \quad
      \begin{cases}
        \oSome ((A',S',J'),\emptyset,O') & \mathrm{if}\, \newNDigests{T\times S\times J}\,(A,S,J)\,u\,u_1 = \oSome (A',S',J')\\
        \oNone & \mathrm{otherwise}
      \end{cases}\\[1ex]

      \initDigests=
      \{((A,S,J),\emptyset,\{i\mapsto(\emptyset,\emptyset)\mid i\in\TIDs\})
      \mid (A,S,J)\in\initNDigests{T\times S\times J}\}\\[1ex]

      \aSemDigest{u, \create\,u_0}\,((A,S,J),C,O)=\\
      \quad\begin{cases}
        \begin{aligned}
          &\Let\,\oSome A_\mathrm{new}=\newNDigests{T}\,A\,u\,u_0\,\In\\
          &\Let\,i_{\mathrm{ego}}=\pTID\,A\,\In\\
          &\Let\,I_\mathrm{new}^*=\{i\in\TIDs\mid
            \pTID\,A_\mathrm{new}\in (\{i\}\cup\mustanc\,i)\}\,\In\\
          &\Let\,C'=C\oplus \{i^*_\mathrm{new}\mapsto S\cap C\, i^*_\mathrm{new}
            \mid i^*_\mathrm{new}\in I^*_\mathrm{new},
            i_\mathrm{ego}\in\mustanc\,i^*_\mathrm{new}\}\,\In\\
          &\oSome ((A',S,J),C',O)
        \end{aligned}
        & \mathrm{if}\ \aSemNDigest{T}{u,\create\,u_0}\,A=\oSome A'\\
        \oNone & \mathrm{otherwise}
      \end{cases}
      \\[1ex]

      \aSemDigest{u, \lock(a)}\,((A_0,S_0,J_0),C_0,O_0)\,((A_1,S_1,J_1),C_1,O_1)=\\
      \quad\Let\,O'=O_0\oplus\left\{\pTID\,A_1\mapsto (C_1,\maycreate\,A_1)\right\}\\
      \quad\begin{cases}
        \oSome((A_0,S',J_0),C_0,O') & \mathrm{if}\,\aSemNDigest{S}{u, \lock(a)}\,S_0 =\oSome S'\\
        \oNone & \mathrm{otherwise}
      \end{cases}\\[1ex]

      \aSemDigest{u, \unlock(a)}\,((A,S,J),C,O)=\\
      \quad\Let\,i_{\mathrm{ego}}=\pTID\,A\,\In\\
      \quad\Let\,I_\mathrm{running}=\{i \in \TIDs\mid\exists i_\mathrm{created}\in \maycreate\, A:
      i_\mathrm{created}\in\{i\}\cup\mustanc\, i\}\setminus J\,\In\\
      \quad\Let\,C'=C\oplus\{i_\mathrm{running}\mapsto C\,i_\mathrm{running}\setminus\{a\}\mid i_\mathrm{running}\in I_\mathrm{running} \}\,\In\\
      \quad\begin{cases}
        \oSome\,((A,S',J),C',O)& \mathrm{if}\, \aSemNDigest{S}{u, \unlock(a)}\,S= \oSome S'\\
        \oNone & \mathrm{otherwise}
      \end{cases}\\[1ex]

      \aSemDigest{\_, act}(((A_0,S_0,J_0),C_0,O_0),\dots,((A_{N-1},S_{N-1}, J_{N-1}),C_{N-1}, O_{N-1})) =\\
      \quad \begin{cases}
        \oSome ((A',S',J'),C_0,O_0) & \text{if } \aSemNDigest{T\times S\times J}{\_,act}((A_0,S_0,J_0),\dots,(A_{N-1},S_{N-1},J_{N-1})) =\\&\quad \oSome (A',S',J')\\
        \oNone & \text{otherwise} \qquad\qquad\qquad\qquad\qquad\qquad \text{(other)}
      \end{cases}
    \end{array}
  \]
  \caption{Creation lockset digest, where $\oplus$ denotes updating a map and missing bindings map to $\Mutexes$.}
  \label{f:descendant-ls-digest}
\end{figure*}
To determine if two program points with the digests explained above may race, we can evaluate the following predicate:

\[
  \begin{array}{l}
    ((A_0,S_0,J_0),C_0,O_0)\, ||^{?}\, ((A_1,S_1,J_1),C_1,O_1) =\\
    \qquad\Let\, i_0 = \pTID\,A_0\,\In\,\\
    \qquad\Let\, i_1 = \pTID\,A_1\,\In\\
    \qquad\begin{cases}
      \false &\mathrm{if\ }
      \begin{array}{l}
        (\exists i_O\in\TIDs:i_O\ne i_1\land\Let\,(C_O,I_O')=O_1\,i_O\,\In\\
        \quad C_O\,i_0\cap S_1\ne\emptyset\;\land(\{i_0\}\cup\mustanc\,i_0)\cap I'_O \neq\emptyset)\;\lor\\
        \color{darkgray}(\exists i_O\in\TIDs:i_O\ne i_0\land\Let\,(C_O,I_O')=O_0\,i_O\,\In\\
        \color{darkgray}\quad C_O\,i_1\cap S_0\ne\emptyset\;\land(\{i_1\}\cup\mustanc\,i_1)\cap I'_O \neq\emptyset)
      \end{array}
      \\
      \top &\mathrm{otherwise}
    \end{cases}
  \end{array}
\]
Consider the first disjunct for the \textsf{false} case.
It checks whether it can be shown that the access of thread $i_1$
must happen after the access of thread $i_0$.
Concretely, $O_1\,i_O=(C_O,I'_O)$ means that $i_1$ observed, through a lock it or one
of its ancestors acquired, an action of a thread with thread \emph{id} $i_O$.
The condition
$({\{i_0\}\cup\mustanc\,i_0})\cap I'_O\ne\emptyset$ establishes that this
observation occurred after $i_O$ had created either $i_0$ or one of its
mandatory ancestors.
If $C_O\,i_0\cap S_1\ne\emptyset$ there is some mutex,
say $m$, that is held at the access by $i_1$ and that was held
throughout the execution of $i_0$.
Thus, we obtain that $i_0$ or one of its mandatory ancestors must have been
created before the access in $i_1$ occurred and $m$ was held throughout $i_0$'s
execution. As $i_0 \neq i_1$, for $i_1$ to succeed in acquiring $m$, it must have been released,
and thus $i_0$ must have completed its execution before $i_1$'s access.
The second, gray disjunct checks the symmetric situation.

\begin{example}
  Consider the program from \cref{f:creation-ls-example} again. At the access to \texttt{g} in $i_1$, the $O$ component of the digest is $$\{\maintid\mapsto(\{\mathit{work}\mapsto\{m\}, i_1\mapsto\emptyset\}, \{i_1,\mathit{work}\})\},$$ where missing bindings of the outer map default to $(\emptyset,\emptyset)$ and missing bindings of the inner map default to $\Mutexes$. Together with $S=\{m\}$, this suffices to rule out a data race with any program point having thread \emph{id} $\mathit{work}$.
\end{example}

\begin{proposition}
  Provided the product digest $\NDigests{T\times S\times J}$ is admissible,
  the digest given above is admissible and its definition for $||^?$ is sound.
\end{proposition}
\begin{proof}
  Let $\alpha_B$ be the abstraction function of the base product digest.
  Define
  \[
    \alpha_\Digests\,t=(\alpha_B\,t,C(t),O(t))
  \]
  inductively along the creation-extended ego lane of $t$. Initially, $C$
  maps every thread \emph{id} to $\Mutexes$ (represented by the empty finite
  map), and $O$ maps every thread \emph{id} to
  $(\emptyset,\emptyset)$. A create intersects the entries of $C$ selected
  by $I_\mathrm{new}^*$ with the current lockset. An $\unlock(a)$ removes
  $a$ from the entries selected by $I_\mathrm{running}$. If a $\lock(a)$
  observes a trace $s$ ending in the matching unlock, it stores
  $(C(s),\maycreate A_s)$ in the $O$ entry for $\pTID A_s$, where
  $\alpha_B\,s=(A_s,S_s,J_s)$. All other actions leave $C$ and $O$
  unchanged. A new thread starts with the initial values of both maps.

  We verify
  \labelcref{def:initSound,def:SemSound,def:NewSound,def:CreateNewConsistent}.
  For requirement~\labelcref{def:initSound}, an initial trace contains no
  creation, unlock, or observation. Hence
  \[
    \alpha_\Digests t
      =((A,S,J),\emptyset,
        \{i\mapsto(\emptyset,\emptyset)\mid i\in\TIDs\})
  \]
  for some $(A,S,J)\in\initNDigests{T\times S\times J}$. Conversely,
  every initial base digest is represented by an initial trace by
  requirement~\labelcref{def:initSound} for the base digest. This proves
  equality with the definition of $\initDigests$.

  For requirement~\labelcref{def:SemSound}, a concrete result $\oNone$ is
  below every abstract result, so it remains to consider successful
  concrete operations. First let $e=(u,\create\,u_0,v)$,
  $\semT{e}(t)=\oSome t'$, and
  \[
    \alpha_\Digests t=((A,S,J),C,O).
  \]
  Base admissibility and flatness of the option order give
  \[
    \aSemNDigest{T}{u,\create\,u_0}A=\oSome A',
    \qquad \alpha_B\,t'=(A',S,J).
  \]
  Admissibility and the componentwise definition of the base product make
  the match $\newNDigests{T}A\,u\,u_0=\oSome A_\mathrm{new}$ total. Let
  $I_\mathrm{new}^*$ be the set induced by $A_\mathrm{new}$ in the
  right-hand side and let
  \[
    C'=C\oplus
      \{i\mapsto S\cap C\,i
        \mid i\in I_\mathrm{new}^*,\ \pTID A\in\mustanc i\}.
  \]
  Exactly these entries acquire the current creation constraint; the create
  performs no observation. Thus
  $\alpha_\Digests t'=((A',S,J),C',O)$, which is the result of the abstract
  create transfer.

  Next suppose
  $\semT{(u,\lock(a),v)}(t_0,t_1)=\oSome t'$, where $t_1$ ends in the
  observed $\unlock(a)$, and write
  \[
    \alpha_\Digests t_k=((A_k,S_k,J_k),C_k,O_k)
    \quad(k\in\{0,1\}).
  \]
  The componentwise base transfer and base admissibility yield
  $\aSemNDigest{S}{u,\lock(a)}S_0=\oSome S'$ and leave the ego thread's
  $A_0$ and $J_0$ unchanged. The lock changes no creation constraint, and
  its mutex-order edge observes precisely $t_1$. Therefore, with
  \[
    O'=O_0\oplus
      \{\pTID A_1\mapsto(C_1,\maycreate A_1)\},
  \]
  we have $\alpha_\Digests t'=((A_0,S',J_0),C_0,O')$, exactly as specified
  by the abstract lock transfer.

  For a successful $\unlock(a)$, let
  $\alpha_\Digests t=((A,S,J),C,O)$. Base admissibility gives
  $\aSemNDigest{S}{u,\unlock(a)}S=\oSome S'$. The set
  \[
    I_\mathrm{running}
      =\{i\in\TIDs\mid\exists i_\mathrm{created}\in\maycreate A:
          i_\mathrm{created}\in\{i\}\cup\mustanc i\}\setminus J
  \]
  contains exactly the possible descendants not known to have been joined.
  Unlocking $a$ invalidates the uninterrupted-holding condition for those
  entries and performs no observation. Hence, for
  \[
    C'=C\oplus
      \{i\mapsto C\,i\setminus\{a\}\mid i\in I_\mathrm{running}\},
  \]
  we obtain $\alpha_\Digests t'=((A,S',J),C',O)$, the result of the
  abstract unlock transfer.

  Finally, for any other successful $n$-ary action, base admissibility and
  flatness give the exact base successor $B'$ from the base digests
  $B_0,\ldots,B_{n-1}$. Such an action changes neither auxiliary component,
  so the successor is $(B',C_0,O_0)$, agreeing with the ``other'' transfer.
  This proves requirement~\labelcref{def:SemSound}.

  For requirement~\labelcref{def:NewSound}, suppose
  $\alpha_\Digests t=(B,C,O)$ and $\new\;t\;u\;u_1=\oSome t'$. Base
  admissibility and flatness yield
  \[
    \newNDigests{T\times S\times J}B\,u\,u_1=\oSome B',
    \qquad \alpha_B\,t'=B'.
  \]
  A new thread has not created a descendant, invalidated a creation
  lockset, or observed another thread. Consequently
  \[
    \alpha_\Digests t'
      =(B',\emptyset,
        \{i\mapsto(\emptyset,\emptyset)\mid i\in\TIDs\}),
  \]
  precisely the result of $\newDigests$. The unsuccessful concrete case
  follows from the minimality of $\oNone$.

  Finally, suppose the premise of
  requirement~\labelcref{def:CreateNewConsistent} holds for $(B,C,O)$.
  Then the corresponding base create transfer succeeds. Requirement
  \labelcref{def:CreateNewConsistent} for the base product yields
  $\newNDigests{T\times S\times J}B\,u\,u_1=\oSome B'$ for some $B'$.
  The displayed definition of $\newDigests$ consequently returns
  \[
    \oSome(B',\emptyset,
      \{i\mapsto(\emptyset,\emptyset)\mid i\in\TIDs\})\ne\oNone.
  \]
  This proves requirement~\labelcref{def:CreateNewConsistent} and hence
  admissibility.

  We now turn to the proof of soundness of $||^?$:
  W.l.o.g., assume the first disjunct of the $\false$ case holds. Consider $t_0\in\gamma_\Digests((A_0,S_0,J_0),C_0,O_0)$ and $t_1\in\gamma_\Digests((A_1,S_1,J_1),C_1,O_1)$.

  Consider some thread with \emph{id} $i_O$, from which the thread at the end of $t_1$ has observed $(C_O,I'_O):=O_1\,i_O$. Since $(\{i_0\}\cup\mustanc\,i_0)\cap I'_O \neq\emptyset$, threads associated with $i_0$ have already been transitively created. Choose $s\in C_O\,i_0\cap S_1$. By the meaning of $C_O\,i_0$, thread $i_O$ holds $s$ throughout the execution of every relevant thread represented by $i_0$. At the same time, $s\in S_1$ means that the ego thread represented by $i_1$ holds $s$ at the access in $t_1$. Since $i_O\ne i_1$, these are different threads. Mutual exclusion orders the corresponding critical sections by $\to_s$, so the accesses cannot be bidirectionally trace-compatible.\qed
\end{proof}

\newcommand{\unreach}{$\texttt{unreach-call}$\ }
\newcommand{\norace} {$\texttt{no-data-race}$\ }

\section{Experimental Setup}\label{app:experiments}
\subsection{Benchmark Subjects}
Our benchmark sets follow the conventions of the \textsc{SV-Comp} software-verification competition.
Each benchmark instance comprises source code together with at least one property to verify:

\begin{itemize}
 \item \norace, which evaluates to \textbf{true} if no data race occurs in any execution of the program;
 \item \unreach, which evaluates to \textbf{true} if the special function \texttt{reach\_error()} is unreachable in every execution of the program; and
 \item other properties, which are not relevant to the benchmarks considered here.
\end{itemize}

The possible verdicts are \textbf{true}, \textbf{false}, \textbf{unknown}, and \textbf{error}. An analyzer that cannot
decide a property conclusively usually returns \textbf{unknown}. Because \textsc{SV-Comp} provides no direct way to
specify invariants, we encode them using program guards and the \unreach property.
Although these benchmarks are formally \unreach tests, they let us check whether an analyzer handles synchronization features correctly.

We use a consistent naming scheme for the benchmark programs. Each filename contains four components separated by underscores.
They identify, in order, the challenge (``b'' for ``barrier,'' ``cl'' for ``creation lockset,''
``dl'' for ``descendant lockset,'' and ``o'' for ``once''),
the property (``n'' for \norace and ``u'' for \unreach), the expected verdict (``true'' or ``false''),
and a number that uniquely identifies the benchmark.

\subsection{Soundness and Precision}

For both \norace and \unreach, we classify verdicts as follows:

\newcommand{\chec}{$\textcolor{teal}\checkmark$}
\newcommand{\parti}{ $\textcolor{red}{\textbf{-}}$ }
\newcommand{\nono}{ $\textcolor{red}{\textbf{+}}$ }
\begin{center}
\begin{tabular}{llcc}
\toprule
\multicolumn{2}{c}{} & \multicolumn{2}{c}{analysis verdict} \\
\cmidrule(lr){3-4}
& & \textbf{TRUE} & \textbf{FALSE} \\
\midrule
\multirow{2}{*}{ground truth }
  & \textbf{TRUE }  & \chec & \textcolor{orange}{ false positive} \\
  & \textbf{FALSE } & \textcolor{red}{ false negative} & \chec \\
\bottomrule
\end{tabular}
\end{center}

\noindent Under this formulation, a \textcolor{orange}{false positive} is a \textbf{FALSE} verdict when
\textbf{TRUE} is expected and indicates imprecision. A \textcolor{red}{false negative} is a \textbf{TRUE} verdict when
\textbf{FALSE} is expected and indicates unsoundness.
\textit{Because abstract interpretation overapproximates program behavior, an abstract interpreter generally cannot substantiate a
\textbf{FALSE} verdict when \textbf{FALSE} is expected and therefore usually returns \textbf{UNKNOWN}.}

\subsection{Comparative Benchmark Baseline}

We selected the baseline tools from \textsc{SV-Comp} participants that performed well
in concurrent-program verification categories. For each tool, we used the parameters specified by its configuration
for the relevant categories. \textsc{RacerF} is an abstract interpreter that supports only the \norace property, so
we did not evaluate it on \unreach.

\begin{center}
\small
\begin{tabularx}{\textwidth}{@{}l r >{\raggedright\arraybackslash\ttfamily}X@{}}
\toprule
\textbf{Tool} & \textbf{Version} & \normalfont\textbf{Configuration} \\
\midrule
\textsc{Racerf} & 2.2 & ./racerf-sv.py -machdep gcc\_x86\_64 \$FILE \\
\textsc{ESBMC} \textit{(base)} & 7.7.0 & ./esbmc --no-div-by-zero-check --force-malloc-success --state-hashing --add-symex-value-sets --no-align-check --k-step 2 --floatbv --unlimited-k-steps --no-vla-size-check \$FILE --64 --no-por --context-bound 3 --witness-output witness.graphml --no-pointer-check --no-bounds-check --incremental-bmc \\
\textsc{ESBMC} \textit{(no-data-race)} & & --data-races-check --no-assertions \\
\textsc{ESBMC} \textit{(unreach-call)} & & --enable-unreachability-intrinsic \\
\textsc{CPAChecker} & 4.2.2 & ./cpachecker --svcomp26 --heap 10000M --benchmark --timelimit 900 s --stats --spec \$PROPERTY --64 \$FILE \\
\textsc{U. Gemcutter} & 0.3.1 & ./Ultimate.py --spec \$PROPERTY --file \$FILE --full-output --architecture 64bit \\
\textsc{Dartagnan} & 4.1.0 & ./Dartagnan-SVCOMP.sh \$PROPERTY \$FILE \\
\bottomrule
\end{tabularx}
\end{center}

\subsection{Benchmark Execution Times}
\noindent We ran the benchmarks on a machine with a 2.5\,GHz Intel Core i7-11700 CPU, 64\,GB of RAM, and Ubuntu 26.04.
Each tool invocation was limited to 5 minutes and 20\,GB of RAM.

{\scriptsize
\setlength{\tabcolsep}{4pt}
\begin{longtable}{l | r | r r r r r r}
\caption{Tool execution times (ms) per input file.} \label{tab:tool_runtimes} \\
\toprule
\textbf{Input File}
      & \rotatebox{90}{\textit{lines of code} }
      & \rotatebox{90}{\textsc{CPAchecker} }
      & \rotatebox{90}{\textsc{Dartagnan} }
      & \rotatebox{90}{\textsc{ESBMC} }
      & \rotatebox{90}{\textsc{RacerF} }
      & \rotatebox{90}{\textsc{UGemCutter} }
      & \rotatebox{90}{\textsc{Goblint} }
      \\
\midrule
\endfirsthead
\multicolumn{8}{c}{\small\bfseries Table \thetable\ continued from previous page} \\
\toprule
\textbf{Input File}
      & \rotatebox{90}{\textit{lines of code} }
      & \rotatebox{90}{\textsc{CPAchecker} }
      & \rotatebox{90}{\textsc{Dartagnan} }
      & \rotatebox{90}{\textsc{ESBMC} }
      & \rotatebox{90}{\textsc{RacerF} }
      & \rotatebox{90}{\textsc{UGemCutter} }
      & \rotatebox{90}{\textsc{Goblint} }
      \\\midrule
\endhead
\bottomrule
\multicolumn{8}{r}{\small Continued on next page} \\
\endfoot
\bottomrule
\endlastfoot
\texttt{b\_n\_false\_1.c}  & 36     & 3013 & 2144 & 1337 & 2941 & 11543 & 82 \\
\texttt{b\_n\_false\_2.c}  & 30     & 2676 & 2017 & 1199 & 3371 & 9471 & 71 \\
\texttt{b\_n\_false\_3.c}  & 29     & 2530 & 2152 & 851  & 2852 & 8991 & 65 \\
\texttt{b\_n\_false\_4.c}  & 33     & 2551 & 2217 & 805  & 2754 & 10669 & 66 \\
\texttt{b\_n\_false\_5.c}  & 34     & 2867 & 2134 & 798  & 2677 & 9867 & 64 \\
\texttt{b\_n\_true\_1.c}   & 35     & 2533 & 2115 & 1316 & 2656 & 9235 & 65 \\
\texttt{b\_n\_true\_2.c}   & 40     & 2532 & 3275 & 2299 & 2891 & 9215 & 74 \\
\texttt{b\_n\_true\_3.c}   & 39     & 2645 & 2161 & 953  & 2851 & 10709 & 93 \\
\texttt{b\_n\_true\_4.c}   & 35     & 2486 & 2423 & 851  & 2846 & 8933 & 59 \\
\texttt{b\_n\_true\_5.c}   & 31     & 2616 & 2398 & 894  & 2708 & 9389 & 56 \\
\texttt{b\_n\_true\_6.c}   & 31     & 2342 & 2239 & 795  & 2788 & 9473 & 65 \\
\texttt{b\_n\_true\_7.c}   & 30     & 2582 & 2475 & 901  & 2812 & 9133 & 59 \\
\texttt{b\_u\_false\_1.i}  & 1420   & 2983 & 2040 & 720  & - & 11506 & 53 \\
\texttt{b\_u\_false\_2.i}  & 1418   & 2736 & 1646 & 720  & - & 13664 & 53 \\
\texttt{b\_u\_false\_3.i}  & 1431   & 3696 & 1655 & 797  & - & 13042 & 76 \\
\texttt{b\_u\_false\_4.i}  & 1422   & 2948 & 1701 & 710  & - & 12700 & 62 \\
\texttt{b\_u\_false\_5.i}  & 1417   & 3037 & 3195 & 797  & - & 12433 & 51 \\
\texttt{b\_u\_false\_6.i}  & 1411   & 2718 & 1762 & 772  & - & 12139 & 47 \\
\texttt{b\_u\_true\_1.i}   & 1415   & 3051 & 1684 & 711  & - & 12992 & 46 \\
\texttt{b\_u\_true\_2.i}   & 1427   & 3376 & 1619 & 701  & - & 14886 & 50 \\
\texttt{b\_u\_true\_3.i}   & 1429   & 2841 & 1667 & 812  & - & 13995 & 56 \\
\texttt{b\_u\_true\_4.i}   & 1424   & 2892 & 3036 & 860  & - & 13357 & 84 \\
\texttt{b\_u\_true\_5.i}   & 1425   & 3005 & 1677 & 882  & - & 12015 & 81 \\
\texttt{b\_u\_true\_6.i}   & 1427   & 2774 & 1594 & 729  & - & 14324 & 59 \\
\texttt{b\_u\_true\_7.i}   & 1427   & 2902 & 1687 & 1112 & - & 13406 & 52 \\
\texttt{b\_u\_true\_8.i}   & 1418   & 3885 & 1612 & 706  & - & 13195 & 46 \\
\texttt{cl\_n\_false\_1.i}  & 39     & 2483 & 2336 & 1124 & 1569 & 9532 & 73 \\
\texttt{cl\_n\_false\_2.i}  & 23     & 2420 & 3368 & 2148 & 1577 & 10176 & 65 \\
\texttt{cl\_n\_false\_3.c}  & 30     & 2594 & 2456 & 845  & 3650 & 10868 & 61 \\
\texttt{cl\_n\_false\_4.c}  & 32     & 2508 & 2863 & 1071 & 3519 & 10238 & 59 \\
\texttt{cl\_n\_false\_5.c}  & 35     & 3506 & 2527 & 992  & 3064 & 10503 & 60 \\
\texttt{cl\_n\_false\_6.c}  & 28     & 2592 & 2550 & 1613 & 2998 & 10024 & 88 \\
\texttt{cl\_n\_false\_7.c}  & 32     & 2626 & 2587 & 1775 & 2804 & 9673 & 64 \\
\texttt{cl\_n\_false\_8.c}  & 23     & 2468 & 2447 & 843  & 2931 & 10829 & 51 \\
\texttt{cl\_n\_false\_9.c}  & 24     & 2590 & 2666 & 1079 & 2778 & 10429 & 55 \\
\texttt{cl\_n\_false\_10.c} & 29     & 2695 & 2132 & 1699 & 2621 & 11228 & 60 \\
\texttt{cl\_n\_false\_11.c} & 39     & 2823 & 2633 & 7174 & 2726 & 10037 & 60 \\
\texttt{cl\_n\_false\_12.c} & 39     & 2930 & 3026 & 3657 & 2120 & 10191 & 62 \\
\texttt{cl\_n\_false\_13.c} & 888    & 2761 & $\infty$ & $\infty$ & 2678 & 10254 & 65 \\
\texttt{dl\_n\_false\_14.i} & 888    & 2367 & 3895 & 1657 & 1501 & 9486 & 48 \\
\texttt{dl\_n\_false\_15.i} & 20     & 2382 & 2923 & 1531 & 1399 & 8840 & 71 \\
\texttt{dl\_n\_false\_16.c} & 21     & 2806 & 2393 & 1959 & 2543 & 10196 & 68 \\
\texttt{dl\_n\_false\_17.c} & 27     & 3216 & 2282 & 827  & 3071 & 10443 & 90 \\
\texttt{dl\_n\_false\_18.c} & 26     & 2936 & 2211 & 1077 & 2648 & 12272 & 63 \\
\texttt{dl\_n\_false\_19.c} & 893    & 2912 & 3663 & 945  & 2683 & 10543 & 55 \\
\texttt{dl\_n\_false\_20.c} & 22     & 2586 & 2205 & 848  & 1867 & 9936 & 54 \\
\texttt{dl\_n\_false\_21.c} & 24     & 2468 & 3647 & 871  & 1859 & 9907 & 54 \\
\texttt{dl\_n\_false\_22.c} & 30     & 2623 & 2653 & 8510 & 3100 & 11235 & 57 \\
\texttt{dl\_n\_false\_23.c} & 23     & 2486 & 2246 & 1839 & 3224 & 9542 & 54 \\
\texttt{dl\_n\_false\_24.c} & 26     & 3017 & 2246 & 818  & 2855 & 11450 & 59 \\
\texttt{dl\_n\_false\_25.c} & 893    & 2527 & 2408 & 13061& 3184 & 9579 & 64 \\
\texttt{cl\_n\_true\_1.c}   & 888    & 2337 & 2363 & 1432 & 3405 & 9468 & 61 \\
\texttt{cl\_n\_true\_2.c}   & 34     & 2431 & 2446 & 5002 & 3083 & 9542 & 61 \\
\texttt{cl\_n\_true\_3.c}   & 25     & 2309 & 3457 & 1067 & 2851 & 9703 & 54 \\
\texttt{cl\_n\_true\_4.c}   & 30     & 2361 & 2580 & 1049 & 3285 & 9335 & 65 \\
\texttt{cl\_n\_true\_5.c}   & 39     & 2527 & 2371 & 9480 & 2010 & 9301 & 89 \\
\texttt{dl\_n\_true\_6.c}   & 23     & 2248 & 2196 & 824  & 2547 & 10063 & 54 \\
\texttt{dl\_n\_true\_7.c}   & 28     & 2455 & 2473 & 982  & 1857 & 9567 & 62 \\
\texttt{dl\_n\_true\_8.c}   & 28     & 2260 & 2152 & 1440 & 2547 & 9702 & 52 \\
\texttt{dl\_n\_true\_9.c}   & 23     & 2342 & 2210 & 970  & 2843 & 9225 & 59 \\
\texttt{dl\_n\_true\_10.i}  & 889    & 2477 & 3555 & 869  & 1632 & 8650 & 49 \\
\texttt{dl\_n\_true\_11.i}  & 28     & 2349 & 2298 & 898  & 1546 & 9943 & 47 \\
\texttt{o\_n\_false\_1.c}  & 26     & 2856 & 2581 & 787  & 1955 & 10821 & 55 \\
\texttt{o\_n\_false\_2.c}  & 47     & 2696 & 2217 & 1065 & 3359 & 9970 & 61 \\
\texttt{o\_n\_false\_3.c}  & 59     & 2731 & 2264 & 796  & 2742 & 11350 & 62 \\
\texttt{o\_n\_false\_4.c}  & 30     & 2450 & 2149 & 841  & 1768 & 9740 & 57 \\
\texttt{o\_n\_false\_5.c}  & 35     & 2442 & 2187 & 868  & 2627 & 9984 & 57 \\
\texttt{o\_n\_false\_6.c}  & 28     & 2470 & 2139 & 842  & 1788 & 9960 & 60 \\
\texttt{o\_n\_true\_1.c}   & 34     & 2903 & 3598 & 851  & 2723 & 9789 & 68 \\
\texttt{o\_n\_true\_2.c}   & 63     & 2424 & 2197 & 905  & 2936 & 9641 & 74 \\
\texttt{o\_n\_true\_3.c}   & 29     & $\infty$ & 3171 & $\infty$ & 2874 & 10716 & 61 \\
\texttt{o\_n\_true\_4.c}   & 28     & 2427 & 3590 & 944  & 3066 & 9132 & 60 \\
\texttt{o\_n\_true\_5.c}   & 41     & 2349 & 2174 & 840  & 2860 & 9391 & 59 \\
\texttt{o\_n\_true\_6.c}   & 40     & 2274 & 2171 & 881  & 2947 & 9371 & 61 \\
\texttt{o\_n\_true\_7.c}   & 48     & 2512 & 3876 & 1985 & 2996 & 10592 & 74 \\
\texttt{o\_n\_true\_8.c}   & 27     & 2420 & 2185 & 949  & 2904 & 10061 & 56 \\
\texttt{o\_n\_true\_9.c}   & 27     & 2428 & 2125 & 1042 & 1747 & 10613 & 79 \\
\texttt{o\_u\_false\_1.i}  & 938    & 2346 & 1669 & 752  & - & 12889 & 55 \\
\texttt{o\_u\_false\_2.i}  & 926    & 3504 & 2796 & 768  & - & 11322 & 65 \\
\texttt{o\_u\_false\_3.i}  & 924    & 2523 & 2232 & 1445 & - & 11151 & 46 \\
\texttt{o\_u\_false\_4.i}  & 925    & 2743 & 2099 & 894  & - & 13314 & 54 \\
\texttt{o\_u\_false\_5.i}  & 928    & 2393 & 1741 & 755  & - & 11400 & 53 \\
\texttt{o\_u\_true\_1.i}   & 930    & 2328 & 1758 & 888  & - & 12621 & 48 \\
\texttt{o\_u\_true\_2.i}   & 1307   & 2518 & 2400 & 1007 & - & 16589 & 60 \\
\texttt{o\_u\_true\_3.i}   & 1302   & 2752 & 2192 & 852  & - & 11251 & 47 \\
\texttt{o\_u\_true\_4.i}   & 1297   & 2191 & 1730 & 1691 & - & 13574 & 72 \\
\texttt{o\_u\_true\_5.i}   & 924    & 2350 & 1652 & 794    & - & 12891 & 50 \\
\texttt{o\_u\_true\_6.i}   & 925    & 2405 & 1832 & 818   & - & 12031 & 75 \\
\texttt{o\_u\_true\_7.i}   & 925    & 2397 & 1731 & 1075 & - & 12693 & 54 \\
\texttt{o\_u\_true\_8.i}   & 928    & 2265 & 1786 & 793    & - & 12433 & 45 \\
\end{longtable}
}

\subsection{Soundness}

To assess soundness, we use \norace and \unreach benchmarks whose expected verdict is \textbf{FALSE} and compare each tool's
result with the expected verdict. A \textbf{TRUE} verdict in these cases misses a data race and is therefore a false negative;
a single such verdict demonstrates unsound behavior. A \textbf{FALSE} or \textbf{UNKNOWN} verdict is consistent with soundness,
but an \textbf{UNKNOWN} verdict does not establish that the property is false.
\textit{\textsc{Dartagnan} and \textsc{Ultimate Gemcutter} recognize occurrences of barriers and \pthreadOnce but return
\texttt{ERROR} and \texttt{UNKNOWN}, respectively. These inconclusive results provide no evidence of unsoundness or precision.
\textsc{RacerF} does not support \unreach and is therefore not evaluated for that property. Because \textsc{Goblint} is an
overapproximating abstract interpreter, it returns \textbf{UNKNOWN}, rather than \textbf{FALSE}, for the false \unreach properties.
It could nevertheless return \textbf{TRUE}; such a result would be a false negative.}

\subsubsection{Benchmark Results}
\begin{center}

{\scriptsize
\setlength{\tabcolsep}{4pt}
\begin{longtable}{l c c c c c c}
\caption{Unsoundness experiment: \textcolor{red}{TRUE} verdicts for \norace or \unreach when FALSE is expected are considered missed data races,
and indicate unsound behaviour for this tool in this test.} \label{tab:tool_verdicts_false} \\
\toprule
\textbf{Input File}
  & %
  {\textsc{CPAchecker} }
  & %
  {\textsc{Dartagnan} }
  & %
  {\textsc{ESBMC} }
  & %
  {\textsc{RacerF} }
  & %
  {\textsc{UGemCutter} }
  & %
  {\textsc{Goblint} }
  \\
\midrule
\endfirsthead
\multicolumn{7}{c}{\small\bfseries Table \thetable\ continued from previous page} \\
\toprule
\textbf{Input File} & \textsc{CPAchecker} & \textsc{Dartagnan} & \textsc{ESBMC} & \textsc{RacerF} & \textsc{UGemCutter} & \textsc{Goblint} \\
\midrule
\endhead
\bottomrule
\multicolumn{7}{r}{\small Continued on next page} \\
\endfoot
\bottomrule
\endlastfoot
\texttt{b\_n\_false\_1.c}  &  FALSE  &  ERROR  &  FALSE  &  FALSE  &  UNKNOWN  &  UNKNOWN \\
\texttt{b\_n\_false\_2.c}  &  FALSE  &  ERROR  &  FALSE  &  FALSE  &  UNKNOWN  &  UNKNOWN \\
\texttt{b\_n\_false\_3.c}  &  FALSE  &  ERROR  &  FALSE  &  FALSE  &  UNKNOWN  &  UNKNOWN \\
\texttt{b\_n\_false\_4.c}  &  FALSE  &  ERROR  &  FALSE  &  FALSE  &  UNKNOWN  &  UNKNOWN \\
\texttt{b\_n\_false\_5.c}  &  FALSE  &  ERROR  &  FALSE  &  FALSE  &  UNKNOWN  &  UNKNOWN \\
\texttt{b\_u\_false\_1.i}  &  FALSE  &  ERROR  &  ERROR  &  -  &  UNKNOWN  &  UNKNOWN \\
\texttt{b\_u\_false\_2.i}  &  FALSE  &  ERROR  &  ERROR  &  -  &  UNKNOWN  &  UNKNOWN \\
\texttt{b\_u\_false\_3.i}  &  FALSE  &  ERROR  &  ERROR  &  -  &  UNKNOWN  &  UNKNOWN \\
\texttt{b\_u\_false\_4.i}  &  FALSE  &  ERROR  &  ERROR  &  -  &  UNKNOWN  &  UNKNOWN \\
\texttt{b\_u\_false\_5.i}  &  FALSE  &  ERROR  &  ERROR  &  -  &  UNKNOWN  &  UNKNOWN \\
\texttt{b\_u\_false\_6.i}  &  FALSE  &  ERROR  &  ERROR  &  -  &  UNKNOWN  &  UNKNOWN \\
\texttt{cl\_n\_false\_1.i}   &  \textcolor{red}{TRUE}  &  \textcolor{red}{TRUE}  &  \textcolor{red}{TRUE}  &  ERROR  &  \textcolor{red}{TRUE}  &  UNKNOWN \\
\texttt{cl\_n\_false\_2.i}   &  \textcolor{red}{TRUE}  &  \textcolor{red}{TRUE}  &  \textcolor{red}{TRUE}  &  ERROR  &  \textcolor{red}{TRUE}  &  UNKNOWN \\
\texttt{cl\_n\_false\_3.c}   &  FALSE  &  FALSE  &  FALSE  &  FALSE  &  UNKNOWN  &  UNKNOWN \\
\texttt{cl\_n\_false\_4.c}   &  FALSE  &  FALSE  &  FALSE  &  FALSE  &  UNKNOWN  &  UNKNOWN \\
\texttt{cl\_n\_false\_5.c}   &  UNKNOWN  &  FALSE  &  FALSE  &  FALSE  &  UNKNOWN  &  UNKNOWN \\
\texttt{cl\_n\_false\_6.c}   &  UNKNOWN  &  \textcolor{red}{TRUE}  &  \textcolor{red}{TRUE}  &  FALSE  &  UNKNOWN  &  UNKNOWN \\
\texttt{cl\_n\_false\_7.c}   &  FALSE  &  FALSE  &  \textcolor{red}{TRUE}  &  FALSE  &  UNKNOWN  &  UNKNOWN \\
\texttt{cl\_n\_false\_8.c}   &  FALSE  &  FALSE  &  FALSE  &  FALSE  &  UNKNOWN  &  UNKNOWN \\
\texttt{cl\_n\_false\_9.c}   &  FALSE  &  FALSE  &  FALSE  &  FALSE  &  UNKNOWN  &  UNKNOWN \\
\texttt{cl\_n\_false\_10.c}  &  FALSE  &  FALSE  &  FALSE  &  FALSE  &  UNKNOWN  &  UNKNOWN \\
\texttt{cl\_n\_false\_11.c}  &  FALSE  &  FALSE  &  \textcolor{red}{TRUE}  &  FALSE  &  UNKNOWN  &  UNKNOWN \\
\texttt{cl\_n\_false\_12.c}  &  FALSE  &  FALSE  &  FALSE  &  \textcolor{red}{TRUE}  &  UNKNOWN  &  UNKNOWN \\
\texttt{cl\_n\_false\_13.c}  &  FALSE  &  UNKNOWN  &  UNKNOWN  &  FALSE  &  UNKNOWN  &  UNKNOWN \\
\texttt{dl\_n\_false\_14.i}  &  \textcolor{red}{TRUE}  &  \textcolor{red}{TRUE}  &  \textcolor{red}{TRUE}  &  ERROR  &  \textcolor{red}{TRUE}  &  UNKNOWN \\
\texttt{dl\_n\_false\_15.i}  &  \textcolor{red}{TRUE}  &  \textcolor{red}{TRUE}  &  \textcolor{red}{TRUE}  &  ERROR  &  \textcolor{red}{TRUE}  &  UNKNOWN \\
\texttt{dl\_n\_false\_16.c}  &  FALSE  &  FALSE  &  FALSE  &  FALSE  &  UNKNOWN  &  UNKNOWN \\
\texttt{dl\_n\_false\_17.c}  &  FALSE  &  FALSE  &  FALSE  &  FALSE  &  UNKNOWN  &  UNKNOWN \\
\texttt{dl\_n\_false\_18.c}  &  FALSE  &  FALSE  &  FALSE  &  FALSE  &  UNKNOWN  &  UNKNOWN \\
\texttt{dl\_n\_false\_19.c}  &  FALSE  &  FALSE  &  FALSE  &  FALSE  &  UNKNOWN  &  UNKNOWN \\
\texttt{dl\_n\_false\_20.c}  &  UNKNOWN  &  FALSE  &  FALSE  &  \textcolor{red}{TRUE}  &  UNKNOWN  &  UNKNOWN \\
\texttt{dl\_n\_false\_21.c}  &  UNKNOWN  &  \textcolor{red}{TRUE}  &  \textcolor{red}{TRUE}  &  \textcolor{red}{TRUE}  &  UNKNOWN  &  UNKNOWN \\
\texttt{dl\_n\_false\_22.c}  &  FALSE  &  FALSE  &  FALSE  &  FALSE  &  UNKNOWN  &  UNKNOWN \\
\texttt{dl\_n\_false\_23.c}  &  FALSE  &  FALSE  &  FALSE  &  FALSE  &  UNKNOWN  &  UNKNOWN \\
\texttt{dl\_n\_false\_24.c}  &  FALSE  &  FALSE  &  FALSE  &  FALSE  &  UNKNOWN  &  UNKNOWN \\
\texttt{dl\_n\_false\_25.c}  &  UNKNOWN  &  ERROR  &  ERROR  &  FALSE  &  UNKNOWN  &  UNKNOWN \\
\texttt{o\_n\_false\_1.c}  &  \textcolor{red}{TRUE}  &  ERROR  &  \textcolor{red}{TRUE}  &  \textcolor{red}{TRUE}  &  UNKNOWN  &  UNKNOWN \\
\texttt{o\_n\_false\_2.c}  &  \textcolor{red}{TRUE}  &  ERROR  &  \textcolor{red}{TRUE}  &  FALSE  &  UNKNOWN  &  UNKNOWN \\
\texttt{o\_n\_false\_3.c}  &  \textcolor{red}{TRUE}  &  ERROR  &  \textcolor{red}{TRUE}  &  FALSE  &  UNKNOWN  &  UNKNOWN \\
\texttt{o\_n\_false\_4.c}  &  \textcolor{red}{TRUE}  &  ERROR  &  \textcolor{red}{TRUE}  &  \textcolor{red}{TRUE}  &  UNKNOWN  &  UNKNOWN \\
\texttt{o\_n\_false\_5.c}  &  \textcolor{red}{TRUE}  &  ERROR  &  \textcolor{red}{TRUE}  &  FALSE  &  UNKNOWN  &  UNKNOWN \\
\texttt{o\_n\_false\_6.c}  &  \textcolor{red}{TRUE}  &  ERROR  &  \textcolor{red}{TRUE}  &  \textcolor{red}{TRUE}  &  UNKNOWN  &  UNKNOWN \\
\texttt{o\_u\_false\_1.i}  &  \textcolor{red}{TRUE}  &  ERROR  &  ERROR  &  -  &  UNKNOWN  &  UNKNOWN \\
\texttt{o\_u\_false\_2.i}  &  FALSE  &  ERROR  &  ERROR  &  -  &  UNKNOWN  &  UNKNOWN \\
\texttt{o\_u\_false\_3.i}  &  FALSE  &  ERROR  &  ERROR  &  -  &  UNKNOWN  &  UNKNOWN \\
\texttt{o\_u\_false\_4.i}  &  \textcolor{red}{TRUE}  &  ERROR  &  ERROR  &  -  &  UNKNOWN  &  UNKNOWN \\
\texttt{o\_u\_false\_5.i}  &  \textcolor{red}{TRUE}  &  ERROR  &  ERROR  &  -  &  UNKNOWN  &  UNKNOWN \\
\midrule
\multicolumn{1}{l}{\textbf{\color{red}\# false negatives}} & \textbf{13} & \textbf{6} & \textbf{14} & \textbf{6} & \textbf{4} & \textbf{0} \\
\multicolumn{1}{l}{\textbf{barriers}}                      & \textbf{0} & \textbf{0} & \textbf{0} & \textbf{0} & \textbf{0} & \textbf{0} \\
\multicolumn{1}{l}{\textbf{once}}                          & \textbf{9} & \textbf{0} & \textbf{6} & \textbf{3} & \textbf{0} & \textbf{0} \\
\multicolumn{1}{l}{\textbf{\creationLockset}}              & \textbf{2} & \textbf{3} & \textbf{5} & \textbf{1} & \textbf{2} & \textbf{0} \\
\multicolumn{1}{l}{\textbf{\descendantLockset}}            & \textbf{2} & \textbf{3} & \textbf{3} & \textbf{2} & \textbf{2} & \textbf{0} \\
\end{longtable}
}
\end{center}

\subsubsection{Example benchmarks indicating unsound behavior}

To illustrate the unsoundness we found, we present compact benchmarks that the corresponding tools report as data-race-free
even though they contain a data race.

Assertions provide another way to expose unsound behavior. In \textsc{SV-Comp}, assertions are typically encoded using the
\unreach property, which can also reveal false negatives, as one of the examples below shows.

\begin{itemize}
\item \textsc{Racerf} reports this program as data-race-free:\par\vspace*{0.75\baselineskip}

\begin{minted}{c}
// expected verdict: no-data-race false
#include <pthread.h>
#include <stdio.h>
#include <assert.h>

int g = 0;
pthread_once_t once = PTHREAD_ONCE_INIT;

void fun() {
  g++; // RACE
}

void *f(void *p){
  fun();
}
pthread_t id1;

int main(void) {
  pthread_t id;

  pthread_create(&id1, NULL, f, NULL);
  pthread_once(&once, fun);

  return 0;
}
\end{minted}
\par\vspace*{0.75\baselineskip}

\item \textsc{CPAChecker} reports this file as data-race-free:\par\vspace*{0.75\baselineskip}
\begin{minted}{c}
// expected verdict: no-data-race false
extern int __VERIFIER_nondet_int();

#include <pthread.h>

int global = 0;
pthread_mutex_t mutex = PTHREAD_MUTEX_INITIALIZER;
pthread_t id1, id2;

void *t1(void *arg) {
  pthread_mutex_lock(&mutex);
  global++; // RACE!
  pthread_mutex_unlock(&mutex);
  return NULL;
}

// t2 is not protected by mutex locked in main thread,
//   joining may not happen before unlock
void *t2(void *arg) {
  global++; // RACE!
  return NULL;
}

int main(void) {
  pthread_create(&id1, NULL, t1, NULL);
  pthread_mutex_lock(&mutex);
  pthread_create(&id2, NULL, t2, NULL);
  int maybe = __VERIFIER_nondet_int();
  if (maybe) {
    pthread_join(id2, NULL);
  }
  pthread_mutex_unlock(&mutex);
  return 0;
}
\end{minted}
\par\vspace*{0.75\baselineskip}

\item \textsc{ESBMC} and \textsc{CPAChecker} report that the error state in this file is unreachable:\par\vspace*{0.75\baselineskip}
\begin{minted}{c}
// expected verdict: unreach-call false
#include <pthread.h>
#include <stdio.h>
#include <assert.h>

extern void abort(void);
void reach_error() { assert(0); }
// expected: unreach-call is false
#define __VERIFIER_assert(cond) { if((cond)) { reach_error(); abort(); } }

int g = 0;
pthread_once_t once = PTHREAD_ONCE_INIT;

void fun() {
  g++;
}


int main(void) {
  pthread_t id;

  pthread_once(&once, fun);
  pthread_once(&once, fun);

  __VERIFIER_assert(g==1); // g==1 is true here -> reach_error() is reached

  return 0;
}
\end{minted}

\end{itemize}

\subsection{Precision}
To assess precision, we use \norace and \unreach benchmarks whose expected verdict is \textbf{TRUE} and compare each tool's
result with the expected verdict. A \textbf{FALSE} verdict in these cases is a false positive.
\textit{\textsc{Dartagnan} and \textsc{Ultimate Gemcutter} recognize occurrences of barriers and \pthreadOnce but return
\texttt{ERROR} and \texttt{UNKNOWN}, respectively. Neither result demonstrates that the tool can verify the property.
\textsc{RacerF} does not support \unreach and is therefore not evaluated for that property.}

\subsubsection{Benchmark Results}

\begin{center}
{\scriptsize
\setlength{\tabcolsep}{4pt}
\begin{longtable}{l c c c c c c}
\caption{Precision experiment: \textcolor{red}{FALSE} verdicts for \norace or \unreach when TRUE is expected are false positives.} \label{tab:tool_verdicts_true} \\
\toprule
\textbf{Input File}
    & %
    {\textsc{CPAchecker} }
    & %
    {\textsc{Dartagnan} }
    & %
    {\textsc{ESBMC} }
    & %
    {\textsc{RacerF} }
    & %
    {\textsc{UGemCutter} }
    & %
    {\textsc{Goblint} }
    \\
\midrule
\endfirsthead
\multicolumn{7}{c}{\small\bfseries Table \thetable\ continued from previous page} \\
\toprule
\textbf{Input File} & \textsc{CPAchecker} & \textsc{Dartagnan} & \textsc{ESBMC} & \textsc{RacerF} & \textsc{UGemCutter} & \textsc{Goblint} \\
\midrule
\endhead
\bottomrule
\multicolumn{7}{r}{\small Continued on next page} \\
\endfoot
\bottomrule
\endlastfoot
\texttt{b\_n\_true\_1.c}  &  \textcolor{red}{FALSE}  &  ERROR  &  \textcolor{red}{FALSE}  &  \textcolor{red}{FALSE}  &  UNKNOWN  &  TRUE \\
\texttt{b\_n\_true\_2.c}  &  \textcolor{red}{FALSE}  &  ERROR  &  \textcolor{red}{FALSE}  &  \textcolor{red}{FALSE}  &  UNKNOWN  &  TRUE \\
\texttt{b\_n\_true\_3.c}  &  \textcolor{red}{FALSE}  &  ERROR  &  \textcolor{red}{FALSE}  &  \textcolor{red}{FALSE}  &  UNKNOWN  &  TRUE \\
\texttt{b\_n\_true\_4.c}  &  TRUE  &  ERROR  &  TRUE  &  \textcolor{red}{FALSE}  &  UNKNOWN  &  TRUE \\
\texttt{b\_n\_true\_5.c}  &  TRUE  &  ERROR  &  TRUE  &  \textcolor{red}{FALSE}  &  UNKNOWN  &  TRUE \\
\texttt{b\_n\_true\_6.c}  &  TRUE  &  ERROR  &  TRUE  &  \textcolor{red}{FALSE}  &  UNKNOWN  &  TRUE \\
\texttt{b\_n\_true\_7.c}  &  \textcolor{red}{FALSE}  &  ERROR  &  \textcolor{red}{FALSE}  &  \textcolor{red}{FALSE}  &  UNKNOWN  &  TRUE \\
\texttt{b\_u\_true\_1.i}  &  \textcolor{red}{FALSE}  &  ERROR  &  ERROR  &  -  &  UNKNOWN  &  TRUE \\
\texttt{b\_u\_true\_2.i}  &  \textcolor{red}{FALSE}  &  ERROR  &  ERROR  &  -  &  UNKNOWN  &  TRUE \\
\texttt{b\_u\_true\_3.i}  &  \textcolor{red}{FALSE}  &  ERROR  &  ERROR  &  -  &  UNKNOWN  &  TRUE \\
\texttt{b\_u\_true\_4.i}  &  \textcolor{red}{FALSE}  &  ERROR  &  ERROR  &  -  &  UNKNOWN  &  TRUE \\
\texttt{b\_u\_true\_5.i}  &  \textcolor{red}{FALSE}  &  ERROR  &  ERROR  &  -  &  UNKNOWN  &  TRUE \\
\texttt{b\_u\_true\_6.i}  &  \textcolor{red}{FALSE}  &  ERROR  &  ERROR  &  -  &  UNKNOWN  &  TRUE \\
\texttt{b\_u\_true\_7.i}  &  \textcolor{red}{FALSE}  &  ERROR  &  ERROR  &  -  &  UNKNOWN  &  TRUE \\
\texttt{b\_u\_true\_8.i}  &  \textcolor{red}{FALSE}  &  ERROR  &  ERROR  &  -  &  UNKNOWN  &  TRUE \\
\texttt{cl\_n\_true\_1.c}  &  TRUE  &  TRUE  &  TRUE  &  \textcolor{red}{FALSE}  &  UNKNOWN  &  TRUE \\
\texttt{cl\_n\_true\_2.c}  &  TRUE  &  TRUE  &  TRUE  &  TRUE  &  UNKNOWN  &  TRUE \\
\texttt{cl\_n\_true\_3.c}  &  TRUE  &  TRUE  &  TRUE  &  \textcolor{red}{FALSE}  &  UNKNOWN  &  TRUE \\
\texttt{cl\_n\_true\_4.c}  &  TRUE  &  TRUE  &  TRUE  &  \textcolor{red}{FALSE}  &  UNKNOWN  &  TRUE \\
\texttt{cl\_n\_true\_5.c}  &  \textcolor{red}{FALSE}  &  ERROR  &  TRUE  &  TRUE  &  UNKNOWN  &  TRUE \\
\texttt{dl\_n\_true\_6.c}  &  TRUE  &  TRUE  &  TRUE  &  \textcolor{red}{FALSE}  &  UNKNOWN  &  TRUE \\
\texttt{dl\_n\_true\_7.c}  &  TRUE  &  TRUE  &  TRUE  &  TRUE  &  UNKNOWN  &  TRUE \\
\texttt{dl\_n\_true\_8.c}  &  TRUE  &  TRUE  &  TRUE  &  \textcolor{red}{FALSE}  &  UNKNOWN  &  TRUE \\
\texttt{dl\_n\_true\_9.c}  &  TRUE  &  TRUE  &  TRUE  &  \textcolor{red}{FALSE}  &  UNKNOWN  &  TRUE \\
\texttt{dl\_n\_true\_10.i}  &  TRUE  &  TRUE  &  TRUE  &  ERROR  &  TRUE  &  TRUE \\
\texttt{dl\_n\_true\_11.i}  &  TRUE  &  TRUE  &  TRUE  &  ERROR  &  TRUE  &  TRUE \\
\texttt{o\_n\_true\_1.c}  &  TRUE  &  ERROR  &  TRUE  &  \textcolor{red}{FALSE}  &  UNKNOWN  &  TRUE \\
\texttt{o\_n\_true\_2.c}  &  TRUE  &  ERROR  &  TRUE  &  \textcolor{red}{FALSE}  &  UNKNOWN  &  TRUE \\
\texttt{o\_n\_true\_3.c}  &  UNKNOWN  &  ERROR  &  UNKNOWN  &  \textcolor{red}{FALSE}  &  UNKNOWN  &  TRUE \\
\texttt{o\_n\_true\_4.c}  &  TRUE  &  ERROR  &  TRUE  &  \textcolor{red}{FALSE}  &  UNKNOWN  &  TRUE \\
\texttt{o\_n\_true\_5.c}  &  TRUE  &  ERROR  &  TRUE  &  \textcolor{red}{FALSE}  &  UNKNOWN  &  TRUE \\
\texttt{o\_n\_true\_6.c}  &  TRUE  &  ERROR  &  TRUE  &  TRUE  &  UNKNOWN  &  TRUE \\
\texttt{o\_n\_true\_7.c}  &  UNKNOWN  &  ERROR  &  TRUE  &  \textcolor{red}{FALSE}  &  UNKNOWN  &  TRUE \\
\texttt{o\_n\_true\_8.c}  &  TRUE  &  ERROR  &  TRUE  &  \textcolor{red}{FALSE}  &  UNKNOWN  &  TRUE \\
\texttt{o\_n\_true\_9.c}  &  TRUE  &  ERROR  &  TRUE  &  TRUE  &  UNKNOWN  &  TRUE \\
\texttt{o\_u\_true\_1.i}  &  TRUE  &  ERROR  &  ERROR  &  -  &  UNKNOWN  &  TRUE \\
\texttt{o\_u\_true\_2.i}  &  TRUE  &  ERROR  &  ERROR  &  -  &  UNKNOWN  &  TRUE \\
\texttt{o\_u\_true\_3.i}  &  TRUE  &  ERROR  &  ERROR  &  -  &  UNKNOWN  &  TRUE \\
\texttt{o\_u\_true\_4.i}  &  TRUE  &  ERROR  &  ERROR  &  -  &  UNKNOWN  &  TRUE \\
\texttt{o\_u\_true\_5.i}  &  TRUE  &  ERROR  &  ERROR  &  -  &  UNKNOWN  &  TRUE \\
\texttt{o\_u\_true\_6.i}  &  TRUE  &  ERROR  &  ERROR  &  -  &  UNKNOWN  &  TRUE \\
\texttt{o\_u\_true\_7.i}  &  TRUE  &  ERROR  &  ERROR  &  -  &  UNKNOWN  &  TRUE \\
\texttt{o\_u\_true\_8.i}  &  TRUE  &  ERROR  &  ERROR  &  -  &  UNKNOWN  &  TRUE \\
\midrule
\multicolumn{1}{l}{\textbf{\color{red}\# false positives}} & \textbf{13} & \textbf{0} & \textbf{4} & \textbf{20} & \textbf{0} & \textbf{0} \\
\multicolumn{1}{l}{\textbf{barriers}}            & \textbf{12} & \textbf{0} & \textbf{4} & \textbf{7} & \textbf{0} & \textbf{0} \\
\multicolumn{1}{l}{\textbf{once}}                & \textbf{0}  & \textbf{0} & \textbf{0} & \textbf{7} & \textbf{0} & \textbf{0} \\
\multicolumn{1}{l}{\textbf{\creationLockset}}    & \textbf{1}  & \textbf{0} & \textbf{0} & \textbf{3} & \textbf{0} & \textbf{0} \\
\multicolumn{1}{l}{\textbf{\descendantLockset}}  & \textbf{0}  & \textbf{0} & \textbf{0} & \textbf{3} & \textbf{0} & \textbf{0} \\
\end{longtable}
}
\end{center}

\subsubsection{Example benchmarks indicating imprecise behavior}

To demonstrate the kinds of imprecise behavior we found, we present
compact benchmarks that some analyzers report as containing a data race
but that are actually data-race-free. The following is a representative barrier example:

\begin{minted}{c}
// expected verdict: no-data-race true
#include<pthread.h>
#include<stdio.h>
#include<unistd.h>

extern int __VERIFIER_nondet_int();

int g;

pthread_barrier_t barrier;

void* f1(void* ptr) {
    g = 2; // NORACE
    pthread_barrier_wait(&barrier);

    return NULL;
}

int main(int argc, char const *argv[])
{
    int top = __VERIFIER_nondet_int();
    int i = 0;

    pthread_barrier_init(&barrier, NULL, 2);

    pthread_t t1;
    pthread_create(&t1,NULL,f1,NULL);

    pthread_barrier_wait(&barrier);
    g = 3; //NORACE

    return 0;
}
\end{minted}

\fi

\end{document}